\documentclass[11pt,oneside,english,jou]{amsart}
\usepackage[utf8]{inputenc} 
\usepackage[T1]{fontenc}
\usepackage{babel}
\usepackage{verbatim}
\usepackage{mathrsfs}
\usepackage{amstext}
\usepackage{amsthm}
\usepackage{amssymb}
\usepackage{amsfonts}
\usepackage{amsmath}
\usepackage{cancel}

\usepackage{esint}
\usepackage{enumerate}
\usepackage[unicode=true,pdfusetitle,
bookmarks=true,bookmarksnumbered=false,bookmarksopen=false,
breaklinks=false,pdfborder={0 0 1},colorlinks=false]
{hyperref}
\usepackage[margin=2.5cm]{geometry}
\usepackage[foot]{amsaddr}
\usepackage{mathtools}
\usepackage{ dsfont }
\usepackage[shortlabels]{enumitem}
\usepackage{bm}
\usepackage{caption}

\usepackage[markup=underlined,commandnameprefix=always]{changes}

\definechangesauthor[color=blue]{YG}

\makeatletter
\numberwithin{equation}{section}
\numberwithin{figure}{section}
\theoremstyle{plain}
\newtheorem{thm}{\protect\theoremname}[section]
\theoremstyle{definition}
\newtheorem{rem}[thm]{\protect\remarkname}
\theoremstyle{definition}
\newtheorem{defn}[thm]{\protect\definitionname}
\theoremstyle{plain}
\newtheorem{prop}[thm]{\protect\propositionname}
\theoremstyle{plain}

\theoremstyle{plain}
\newtheorem{lem}[thm]{\protect\lemmaname}
\theoremstyle{plain}

\theoremstyle{plain}

\theoremstyle{definition}

\theoremstyle{definition}

\theoremstyle{definition}
\newtheorem{examplex}[thm]{\protect\examplename}
\theoremstyle{definition}

\newtheorem{lemma}[thm]{Lemma}

\newenvironment{example}
{\pushQED{\qed}\examplex}
{\popQED\endexamplex}

\usepackage{tikz}
\usetikzlibrary{shapes,arrows}
\usepackage{verbatim}
\usepackage{amsthm}
\usepackage{amstext}

\DeclareMathOperator{\Leb}{Leb}

\DeclareMathOperator{\supp}{supp}

\DeclareMathOperator{\Var}{Var}

\newcommand{\R}{\mathbb R}
\newcommand{\Z}{\mathbb Z}
\newcommand{\N}{\mathbb N}

\newcommand{\PP}{\mathbb P}

\newcommand{\eps}{\varepsilon}

\newcommand{\mC}{\mathcal{C}}

\newcommand{\mE}{\mathcal{E}}

\newcommand{\E}{\mathbb{E}}

\newcommand{\mdimcor}{\mdim_{\mathrm{cor}}}

\newcommand{\ualdim}{\overline{\dim}_{AL}\,}
\newcommand{\laldim}{\underline{\dim}_{AL}\,}
\newcommand{\aldim}{\dim_{AL}\,}
\newcommand{\mualdim}{\overline{\mdim}_{AL}\,}
\newcommand{\mlaldim}{\underline{\mdim}_{AL}\,}

\newcommand{\hdim}{\dim_H}

\newcommand{\umid}{\overline{\mathrm{mid}}}
\newcommand{\comp}{\mathrm{comp}}

\newcommand{\lmid}{\underline{\mathrm{mid}}}
\newcommand{\midim}{\mathrm{mid}}
\newcommand{\uid}{\overline{\mathrm{id}}}
\newcommand{\lid}{\underline{\mathrm{id}}}
\newcommand{\idim}{\mathrm{id}}

\newcommand{\adim}{\dim_A}
\newcommand{\ld}{\underline{d}}
\newcommand{\ud}{\overline{d}}
\newcommand{\essinf}{\mathrm{essinf}}
\newcommand{\esssup}{\mathrm{esssup}}
\newcommand{\uidimr}{\overline{\mathrm{idimr}}}
\newcommand{\lidimr}{\underline{\mathrm{idimr}}}
\newcommand{\idimr}{\mathrm{idimr}}

\DeclareMathOperator{\mdim}{mdim}

\DeclareMathOperator{\rank}{rank}

\makeatother

\providecommand{\fundprobname}{Fundamental Problem – Achievability}

\providecommand{\conjecturename}{Conjecture}
\providecommand{\corollaryname}{Corollary}
\providecommand{\definitionname}{Definition}
\providecommand{\examplename}{Example}
\providecommand{\lemmaname}{Lemma}
\providecommand{\problemname}{Problem}
\providecommand{\propositionname}{Proposition}
\providecommand{\remarkname}{Remark}
\providecommand{\theoremname}{Theorem}
\providecommand{\taskname}{Task}

\newcommand{\om}{\omega}

\def\Om{\Omega}

\def\N{{\mathbb N}}

\def\C{{\mathcal C}}

\def\Pk{{\mathcal P}}
\def\Ck{{\mathcal C}}

\def\be{\begin{equation}}
	\def\ee{\end{equation}}

\newcommand{\Ek}{{\mathcal E}}
\newcommand{\Fk}{{\mathcal F}}

\title[Optimal compressed sensing]{Optimal compressed sensing for mixing stochastic processes}

\author[Y. Gutman]{Yonatan Gutman$^\dagger$}
\address{$^\dagger$Institute of Mathematics, Polish Academy of Sciences,
	ul.~\'Sniadeckich 8, 00-656 Warszawa, Poland}
\email{gutman@impan.pl}

\author[A. \'{S}piewak]{Adam \'{S}piewak$^\dagger$}

\email{ad.spiewak@gmail.com}

\thanks{YG and A\'S were partially supported by the National Science Centre (Poland) grant 2020/39/B/ST1/02329.}

\begin{document}

\begin{abstract}
		Jalali and Poor introduced an asymptotic framework for compressed sensing of stochastic processes, demonstrating that any rate strictly greater than the mean information dimension serves as an upper bound on the number of random linear measurements required for (universal) almost lossless recovery of $\psi^*$-mixing processes, as measured in the normalized $L^2$ norm. In this work, we show that if the normalized number of random linear measurements is strictly less than the mean information dimension, then almost lossless recovery of a $\psi^*$-mixing process is impossible by any sequence of decompressors. This establishes the mean information dimension as the fundamental limit for compressed sensing in this setting (and, in fact, the precise threshold for the problem). To this end, we introduce a new quantity, related to techniques from geometric measure theory: the correlation dimension rate, which is shown to be a lower bound for compressed sensing of arbitrary stationary stochastic processes.
\end{abstract}


	\maketitle

	\section{Introduction}
\subsection{Compressed sensing for stochastic processes}

The field of \textit{compressed sensing} originated from the foundational work by Candès, Donoho, Romberg, and Tao~\cite{candes2006compressive,candes2006stable,donoho2006compressed,FR13} and others. A central result in the theory~\cite[Theorem 9.12]{FR13} (see also~\cite{candes2006near,candes2006robust}) asserts that, with high probability, any vector \( x \in \mathbb{R}^N \) satisfying the \textit{\( s \)-sparsity} condition — i.e., \( \|x\|_0 := \#\{j : x_j \neq 0\} \leq s \) — can be recovered  with high probability  from $m$ random (Gaussian) linear measurements \( y := \mathbf{A}x \in \mathbb{R}^m \), where \( m \approx s \ln(N/s) \). This recovery is achieved using an \( \ell_1 \)-minimization method known as \textit{basis pursuit}~\cite[§1.4.3]{mallat1999wavelet} (see also~\cite{chen1995atomic}). Leveraging signal sparsity, compressed sensing has since enabled a wide range of applications~\cite{lustig2007sparse,duarte2008single,baraniuk2007compressive,herman2009high}.
However, from a practical perspective, it is advantageous to develop recovery algorithms that are applicable to sources exhibiting more general structural characteristics than sparsity. 
Following the tradition of information theory it is natural to model the source with the help of a  real-valued stationary stochastic processes $\bm{X} :=(X_i)_{i=1}^\infty$\footnote{A possible critique of the stochastic approach is that by its own nature it does not guarantee decompression for \textit{all} source vectors, at best giving almost sure results. For some related uniform results see \cite{mmdimcompress}.}. As measures of structural complexity we will primarily use several quantities: \textit{information dimension rate}, \textit{mean information dimension} and \textit{correlation dimension rate}, all defined in Section \ref{sec: Preliminaries}. The  fundamental role of these quantities is justified by the results presented in the sequel as well as previous results in the literature (e.g. \cite{WV10, jalali2017universal,rezagah2017compression,geiger2019information}). As an example let us recall that the \textbf{upper  mean information dimension} of a stationary stochastic processes $\mathbf{X}$ with distribution $\mu$ is defined as
\begin{equation}\label{eq: umid def}
	\umid(\mathbf{X}) = \lim \limits_{n \to \infty} \limsup \limits_{k \to \infty} \frac{H_\mu([X^n]_k)}{n \log k},
\end{equation}
where $[X^n]_k := \frac{\lfloor k (X_1, \ldots, X_n) \rfloor}{k}$ and $H_\mu(\cdot)$ denotes the Shannon entropy with respect to distribution $\mu$. This quantity measures the linear growth rate of the Rényi information dimension  along finer and finer quantizations of the process. It was introduced by Jalali and Poor \cite{jalali2017universal}\footnote{Jalali and Poor called $\umid(\mathbf{X})$ simply the (upper) information dimension of a process. We adopt the name \textit{mean information dimension} to emphasize the averaging over the dynamics of the stochastic process.} (the original definition is different, but it agrees with the above one by \cite[Lemma 3]{jalali2017universal}). See Section \ref{subsec: mid} for more details. For the sake of illustration let $\comp(\mathbf{X})$, be one of the complexity quantities mentioned above e.g., $\comp(\cdot)=\umid(\cdot)$ (or another relevant complexity quantity). Using this notation one may reformulate the main problem studied by Jalali and Poor in \cite{jalali2017universal} (elaborating on previous  research  such as \cite{JalaliMalekiBaraniuk14,ZhuBaronDuarte15}) in the following way:  \\

\textbf{Fundamental Problem - Achievability.}
	\begin{itemize}
		\item[--] Let $\C$ be a class of (stationary) stochastic process, e.g., ergodic stochastic process or $\psi^*$-mixing stochastic process.
		\item[--] Let   $\bm{X} =(X_i)_{i=1}^\infty$  be a stochastic process belonging to $\C$ whose 
		distribution is denoted by $\mu$ and $\bm{A} :=(A_n)_{n=1}^\infty$ be an i.i.d. stochastic process of Gaussian matrices $A_n \in \R^{m_n \times n}$,  $n=1,2,\ldots$, independent from $\bm{X}$, known as \textbf{compressors} whose 
		distribution is denoted by $\nu$.
		\item[--] Assume $\lim \limits_{n\rightarrow \infty} \frac{m_n}{n}> \comp(\mathbf{X})$.
		\item[--] Can one find a family of Borel maps $F_n :  \R^{m_n} \times  \R^{m_n \times n }\to \R^n$, $n=1,2,\ldots$, known as \textbf{decompressors} so that  \textbf{almost lossless recovery}  holds in the following sense
		\[\frac{1}{\sqrt{n}}\Big\|(X_1, \ldots, X_n) - F_n (A_n (X_1, \ldots, X_n), A_n)  \Big\|_2 \overset{n \to \infty}{\longrightarrow} 0 \text{ in } (\mu\times \nu)-\text{probability?} \]
	Above $\|\cdot\|_2$ denotes the Euclidean norm.
	\end{itemize}

\noindent
Jalali and Poor showed in \cite[Theorem 7]{jalali2017universal} that for $\comp(\cdot)=\umid(\cdot)$, the above question has a positive answer in the class of the $\psi^*$-mixing stochastic processes (thus in particular for i.i.d. processes). Let us recall the definition of $\psi^*$-mixing.

\begin{defn}\label{def: psi star mixing}
	Given a stochastic process $\mathbf{X}$, define for $g \in \N$ the $\mathbf{\psi}^*$\textbf{-mixing coefficient} as
	\[ \psi^*(g) = \sup \frac{\PP(A \cap B)}{\PP(A)\PP(B)}, \]
	where the supremum is taken over all $n \in \N$ and events $A \in \sigma(X_1^n), B \in \sigma(X_{n+g}^\infty)$ such that $\PP(A) > 0$ and $\PP(B) > 0$. A process $\mathbf{X}$ is called $\mathbf{\psi}^*$\textbf{-mixing} if $\lim \limits_{ g \to \infty} \psi^*(g) = 1$.
\end{defn}

Basic examples of $\psi^*$-mixing processes include i.i.d. processes and finite state irreducible aperiodic Markov chains. See \cite{BradleySurvey05} for a comprehensive survey. We prove in Appendix \ref{sec: markov psi star proof} that continuous-state Markov chains with uniformly bounded transition densities are $\psi^*$-mixing.

Returning to the discussion of the Jalali--Poor result \cite[Theorem 7]{jalali2017universal}, their decompressors are given explicitly and produce vectors that match the observed random linear measurements while minimizing a certain empirical entropy functional. In fact, the decompressors in \cite{jalali2017universal} are \textit{universal} in the sense that they are constructed without a prior knowledge of the distribution of $\bm{X}$. Thus the above framework, which incorporates asymptotic analysis of compressed sensing where both the compression matrix and the  input vector are random, results with a considerable extension of the sparsity paradigm, while still providing a universal decompression algorithm\footnote{The paper \cite{jalali2017universal} also contains an extension of the above result to a noisy setting.}.

A natural question which arises is if one may  extend the scope of the above result to a larger class of processes. Indeed it is unknown if the result of \cite{jalali2017universal} stands for the class of all (ergodic) stochastic processes\footnote{Note that a positive result for the class of all stochastic processes was achieved in a weaker setting where one allows decompressors dependent on the distribution of $\bm{X}$ \cite{rezagah2017compression}, i.e. in the non-universal setting. The result is given in terms of $\comp(\cdot) = \idimr(\cdot)$, where $\idimr(\cdot)$ is the information dimension rate (see Section \ref{subsec: mid} and \cite{geiger2019information} for the proof of equality with the rate-distortion dimension employed originally in \cite{rezagah2017compression})}   
Another important question is the question of the so-called \textit{converse}. We achieve a converse under mild technical conditions:

\begin{defn}
	A probability measure $\mu$ on $\R^n$ is said to be \textbf{local dimension regular} if the limit 
	$$
	\lim \limits_{r \to 0} \frac{\log \mu(B^n_2(x,r))}{\log r}
	$$
	exists for $\mu$-a.e. $x\in \R^n$, where $B^n_2(x,r)=\{y\in \R^n:\, \|x-y\|_2 \leq r\}$.
\end{defn}
We will apply this definition to the  measures $\mu_n$ being the distributions of $(X_1,\ldots, X_n)$, referred to as the \textbf{finite-dimensional marginals} of a stochastic process $\bm{X} =(X_i)_{i=1}^\infty$. Our main result is the following.

\begin{thm}[Main Theorem --- Converse for $\psi^*$-mixing stochastic processes]\label{main theorem}
	Let $\bm{X} =(X_i)_{i=1}^\infty$ be a finite variance, stationary, $\psi^*$-mixing stochastic process with local dimension regular finite-dimensional marginals. Consider a sequence $m_n  \in \N$ such that
	$$
	\liminf \limits_{n \to \infty} \frac{m_n}{n} < \umid(\mathbf{X}).
	$$
	Let  $F_n :  \R^{m_n} \times \R^{m_n \times n }\to \R^n$ be a sequence of Borel maps, i.e. an arbitrary family of decompressors. 
	Let $\bm{A} :=(A_n)_{n=1}^\infty$ be an i.i.d. stochastic process of Gaussian matrices $A_n \in \R^{m_n \times n}$,  $n=1,2,\ldots$ (with entries drawn i.i.d from the $N(0,1)$ distribution), independent from $\bm{X}$, with  
	distribution $\nu$. Then 
	\[\frac{1}{\sqrt{n}}\Big\|(X_1, \ldots, X_n) - F_n (A_n (X_1, \ldots, X_n), A_n)  \Big\|_2 
	\]
	does not converge to zero in $(\mu\times \nu)$-probability as $n \to \infty$.
\end{thm}

The proof is given in Section~\ref{sec: laldim main}; it follows directly from Theorem~\ref{thm: laldim main} and Lemma~\ref{lem: aldim idim}.

This result, together with the result of Jalali and Poor \cite[Theorem 7]{jalali2017universal} mentioned above, essentially establishes the upper mean information dimension as the 
\textit{fundamental limit} of compressed sensing of $\psi^*$-mixing stochastic processes. Note that the result is stronger than a converse for the universal compression, as it does not require the compressors to be universal (hence it gives a converse also to the results of \cite{rezagah2017compression} in the class of $\psi^*$-mixing processes). In particular it may be applied to $\psi^*$-mixing Gaussian processes, or i.d.d. sources with mixed discrete-continuous or regular enough fractal distributions - see discussion in Examples \ref{ex: gaussian} and \ref{ex: i.i.d.}. See also Example \ref{ex: sparse} for the analysis of the asymptotically sparse case.

One should note that $\psi^*$-mixing is a strong uniform mixing condition, substantially stronger than the commonly used $\alpha$- (strong), $\beta$-, and $\phi$-mixing conditions; see \cite{BradleySurvey05} for a comprehensive comparison. For example, within the class of stationary Gaussian processes, $\psi^*$-mixing is highly restrictive: it is equivalent to being $m$-dependent for some $m\in \N$, i.e. independence holds for all lags bigger than $m$ (see \cite[Theorem~7.1(3)]{BradleySurvey05}), which in turn is equivalent to the process being a finite-order moving average of an i.i.d.\ sequence of standard Gaussian random variables (see \cite[Theorem~4.4.2]{BrockwellDavis1991}). Thus, in general,   stationary Gaussian autoregressive processes (and more generally stationary  Gaussian ARMA processes) will not be $\psi^*$-mixing. In the same vein Gaussian processes with band-limited or  singular spectral measures are incompatible with $\psi^*$-mixing, as another equivalent characterization of $\psi^*$-mixing in the Gaussian case is the spectral density being of the form $|p(e^{i\lambda})|^2$ for some polynomial~$p$  (see \cite{ibragimov1961spectral}).
 Beyond the Gaussian setting, intermittent dynamical systems i.e., systems which exhibit irregular alternation of phases of apparently periodic and chaotic dynamics (see \cite[Chapter 8.2]{Ott02}) will not be in general $\psi^*$-mixing. In the prototypical example of the Pomeau-Manneville map this can be seen as a consequence of the infinite dependencies for small neighborhoods of the neutral fixed point at $x=0$ (see \cite{Gouezel2004, freitas2016rare}). More generally, metastable systems, i.e., stochastic processes exhibiting multiple almost-invariant regions with rare transitions between them, such as the Curie-Weiss Ising model under Glauber dynamics (see \cite[Chapter~14]{BovierDenHollander2015}) will typically fail to have $\psi^*$-mixing. 
 Our result matches the achievability result of Jalali and Poor \cite{jalali2017universal}, who make the same $\psi^*$-mixing assumption. As far as we know, the converse is new even in the case of i.i.d.\ processes, and dealing with the i.i.d.\ case alone does not appear to significantly simplify the proof, at least with the methods of the present paper. A general class of processes to which the result applies are Markov chains on continuous state spaces with uniformly bounded transition densities with respect to a given reference measure (not necessarily Lebesgue). The $\psi^*$-mixing property holds for arbitrary reference measures, while $\midim(\mathbf{X})$ is easy to calculate for local dimension regular measures, which can include highly fractal measures as elaborated in Example \ref{ex: non dyn}, as well as the canonical examples of i.i.d. processes with mixed discrete-continuous distribution (Example \ref{ex: sparse} below), which have been considered in the context of information theory and analog compression - see e.g. \cite[Theorems 1 and 6]{WV10} or \cite[Corollary~1]{jalali2017universal}. See Proposition~\ref{prop: markov psi star} in Appendix~\ref{sec: markov psi star proof} for the precise statement and proof.

While Theorem~\ref{main theorem} provides a full converse to the achievability problem at the cost of requiring $\psi^*$-mixing and local dimension regularity, we also obtain results under weaker assumptions. Theorem~\ref{thm: mdimcor detail} gives a lower bound on the compression rates in terms of the correlation dimension rate (Definition~\ref{def: mdimcor}) for arbitrary sequences of measures, and hence in particular for arbitrary stochastic processes. The difficulty in applying it lies in the fact that calculating the correlation dimension rate appears to be quite challenging; we do not pursue this problem further in this paper. For us, Theorem~\ref{thm: mdimcor detail} serves as the main technical step for obtaining Theorem~\ref{main theorem}, but it may also be of independent interest. Assuming $\psi^*$-mixing but dropping the local dimension regularity assumption, we obtain an analog of Theorem~\ref{main theorem} with $\midim(\mathbf{X})$ replaced by the mean average local dimension $\mlaldim(\mathbf{X})$ (Definition~\ref{def: mlaldim}) - see Theorem~\ref{thm: laldim main}. This result, too, is primarily a step in the proof of Theorem~\ref{main theorem}, but may be of independent interest as well.

\subsection{Comparison with the Wu-Verd\'u theory}

In recent years there has been a surge in interest in a compressed sensing framework for \textbf{analog signals} modeled by continuous-alphabet discrete-time stochastic processes\footnote{The rigorous passage between continuous-time signals and discrete-time signals is justified by the Nyquist-Shannon sampling theorem (\cite[Chapter 1]{higgins1996sampling}).} (\cite{WV10,donoho2010precise,donoho2011noise,jalali2017universal,rezagah2017compression,mmdimcompress,geiger2019information}). Let us remark that fundamental limits for analog compression have been obtained before, but none of those results apply to the setting of Theorem~\ref{main theorem}. In particular, Wu and Verd\'u \cite{WV10} consider only exact recovery with high probability (i.e. they consider compression schemes with $\PP(X^n \neq \hat{X}^n) < \eps$ for all $n$ large enough). It follows from  \cite[Corollary IX-A.2 and (13)]{mmdimcompress} and \cite[Theorem 9]{geiger2019information} that the information dimension rate $\uidimr(\mathbf{X})$ (see Section \ref{subsec: mid}) is a fundamental limit for the convergence in probability, but only in the case when the recovery function $F_n(y,A)$ is a Lipschitz function of $y$, which is essential for the argument in \cite{mmdimcompress} (see also \cite{GSNewBounds}). As the decompressors appearing in \cite{jalali2017universal, rezagah2017compression} are not even continuous (as they employ quantization), considering discontinuous recovery functions is crucial for applications.

\subsection{The asymptotically sparse case}

As a simple, yet informative application of our results, let us consider the case of i.i.d systems generating asymptotically sparse vectors.
\begin{example}\label{ex: sparse}
	Fix $p \in (0,1)$ and let $\mu_1 = (1-p)\delta_0 + p\Leb|_{[0,1]}$. Set $\mu = \mu_1^\N$ and note that $\mu$ is a distribution of an i.i.d. stochastic process $\bm{X} = (X_1, X_2, \ldots)$ with mixed discrete-continuous distribution. It follows from the Strong Law of Large Numbers that
	\[ \lim \limits_{n \to \infty}\frac{1}{n}\|(X_1, \ldots, X_n)\|_0 = p \text{ almost surely},\]
	hence a typical realization of the process is asymptotically $(pn + o(n))$-sparse. The assumptions of Theorem~\ref{main theorem} are satisfied by $\bm{X}$ and
	\[\midim(\bm{X}) = \idim(X_1) = p\]
	(see Example \ref{ex: i.i.d.} for details). It therefore follows from Theorem~\ref{main theorem} (together with the results of \cite{jalali2017universal, rezagah2017compression}) that the condition
	\[ 	\liminf \limits_{n \to \infty} \frac{m_n}{n} > p \]
	is the precise threshold for the existence of decompressors $F_n$ providing an almost lossless recovery of $(X_1, \ldots, X_n)$ from its random Gaussian measurement $A_n(X_1, \ldots, X_n)$, i.e. satisfying $\lim \limits_{n \to \infty}\frac{1}{\sqrt{n}}\Big\|(X_1, \ldots, X_n) - F_n (A_n, A_n (X_1, \ldots, X_n))  \Big\|_2 \to 0$ in probability.
\end{example}

The above can be compared with more constrained problems of finding the asymptotic thresholds for the problems of recovery of sparse vectors using the $\ell_1$-minimization algorithm. This can be considered in the setting of \textit{uniform recovery} (i.e. for high probability of Gaussian matrices $A$, recovering every $s$-sparse vector $x$ from its measurement $y= Ax$ via $\ell_1$-minimization) and \textit{non-uniform recovery} (i.e. for fixed $s$-sparse vector $x$, recovering it from the measurement $y = Ax$ via $\ell_1$-minimization with a high probability on the draw of a Gaussian matrix $A$). The asymptotic study was performed by Donoho and Tanner \cite{Donoho06a, DonohoTanner05, DonohoTanner05a, DonohoTanner09}. In this case, the thresholds are more complicated, require more measurements and they are not given in closed forms - see \cite{Donoho06a, DonohoTanner09} for more details and \cite[Section 9]{FR13} for a summary.

\subsection{The method and the structure of the paper}

For the proof of Theorem~\ref{main theorem}, we introduce two new complexity measures for stochastic processes and prove lower bounds for compression rates in their terms. The \textit{correlation dimension rate} (denoted $\mdimcor(\mathbf{X})$) gives a lower bound on compression rates for arbitrary stationary processes (see Theorem~\ref{thm: mdimcor detail}). Using the $\psi^*$-mixing condition (see Proposition \ref{prop: restriction}), we then improve this to a lower bound in terms of the \textit{mean average local dimension} (denoted $\mlaldim(\mathbf{X})$), which is easier to compute (see Theorem~\ref{thm: laldim main}). Theorem~\ref{main theorem} finally follows from the fact that $\mlaldim(\mathbf{X})$ coincides with $\midim(\mathbf{X})$ under mild regularity assumptions (see Lemma~\ref{lem: aldim idim}). As these new notions are central to our methods, we include in Section~\ref{subsec: discussion} a discussion of some of their relations to more standard concepts in information theory.

The proof of Theorem \ref{thm: mdimcor detail}, which is our main technical result, is based on combining energy method of \cite{JMSections} with concentration inequalities for high-dimensional Gaussian matrices \cite{JalaliMalekiBaraniuk14,VershyninBook}. It can be seen as an attempt to develop methods of \textit{high-dimensional geometric measure theory}, which can be applied to stochastic processes rather than finite-dimensional measures. The correlation dimension rate can be seen as a dynamical version of the correlation dimension, defined in terms of energy integrals (see Section \ref{subsec: corr dim}). From our point of view, its usefulness stems from the fact that it works well with potential-theoretic (energy) methods of proving projection \cite{HuTaylorProjections}, embedding \cite{BGS20, BGSRegularity} and slicing \cite{JMSections} theorems for random orthogonal projections. See \cite{BGSRegularity} for a discussion of its connections with fundamental limits of lossless compression by random linear maps in a fixed finite dimension. It turns out that for our needs  it is crucial to have a quantitative control on the growth of energies of finite-dimensional marginals of a stochastic process.

The paper is organized as follows. Section \ref{sec: Preliminaries} introduces basic definitions and concepts.  In Section \ref{sec: mdimcor main} we prove Theorem \ref{thm: mdimcor detail} (converse for general sources in terms of the correlation dimension rate), while in Section \ref{sec: laldim main} we prove Theorem \ref{thm: laldim main} (converse for $\psi^*$-mixing processes in terms of the mean average local dimension) and deduce Theorem~\ref{main theorem} from it. The appendices contain auxiliary proofs and additional examples.

	\section{Dimensions: definitions and notation}\label{sec: Preliminaries}


The methods of this paper are strongly based on dimension theory methods from the geometric measure theory and its dynamical counterpart, adapted to the setting of stochastic processes. This section provides a detailed discussion of the necessary apparatus, together with definitions of complexity measures for stochastic processes that will be used throughout the paper (including some relations with more standard notions from information theory and geometric measure theory). Section~\ref{subsec: discussion} discusses some connections of these (mean) dimension notions with entropy.

\subsection{General notation and standing conventions}\label{subsec: conventions}

Throughout the article, all logarithms will be in base $2$ and $\|\cdot\|_2$ will always denote the Euclidean norm on $\R^n$. We shall write $B_2^n(x,r)$ for the closed $r$-ball around $x$ in the Euclidean norm and $B_\infty^n(x,r)$ for the closed ball in the supremum norm. For a linear map $A : \R^n \to \R^k$ we will denote by $\|A\|$ the operator norm of $A$ with respect to Euclidean norms on $\R^n$ and $\R^k$.

By $\Leb_n$ we shall denote the Lebesgue measure on $\R^n$ and by $\alpha(n) = \Leb_n(B_2^n(0,1))$ the volume of a unit $n$-ball, so that $\Leb_n(B_2^n(x, r)) = \alpha(n)r^n$. For a measure $\mu$ and a measurable map $\phi$,  we will denote the transport of $\mu$ by $\phi$ as $\phi \mu$, i.e. $$\phi\mu(A) := \mu(\phi^{-1}{A})$$ for measurable sets $A$.

Given a strictly increasing sequence $n_k$ of natural numbers and two sequences $A_{n_k}$ and $B_{n_k}$ we write $A_{n_k} \lesssim^e B_{n_k}$ to denote that there exists $C \geq 0$ such that inequality $A_{n_k} \leq C^{n_k} B_{n_k}$ holds for every $k \in \N$ large enough (so $B_{n_k}$ bounds $A_{n_k}$ up to an exponential factor). If $C$ and the range of $k$ are allowed to depend on some parameters, this will be indicated in the lower index, e.g. $A_{n_k} \lesssim^e_{M,\delta} B_{n_k}$ means that there exists $C(M, \delta)$ and $k_0(M, \delta)$ such that $A_{n_k} \leq C(M, \delta)^{n_k} B_{n_k}$ for all $k \geq k_0(M, \delta)$.  Throughout the paper, $A_{n_k}$ and $B_{n_k}$ will depend on a given stochastic process. As we will work with a fixed stochastic process~$\mathbf{X}$, we suppres the dependence on $\mathbf{X}$ in the notation.

\subsection{Local dimensions}
\begin{defn}
	Let $\mu$ be a probability measure on $\R^n$. We define the \textbf{lower and upper local dimensions} of $\mu$ at $x \in \supp \mu$ as
	\[ \ld(\mu, x) := \liminf \limits_{r \to 0} \frac{\log \mu(B^n_2(x,r))}{\log r},\ \ud(\mu, x) := \limsup \limits_{r \to 0} \frac{\log \mu(B^n_2(x,r))}{\log r}\]
	and $\ld(\mu, x) = \ud(\mu, x) = 0$ for $x \notin \supp \mu$. If $\ld(\mu, x) = \ud(\mu, x)$, then we denote their common value $d(\mu, x)$ and call it the \textbf{local dimension} of $\mu$ at $x$. The \textbf{lower and upper average local dimensions} of $\mu$ are defined as
	\[ \laldim(\mu) := \int \ld(\mu, x)d\mu(x),\ \ualdim(\mu) := \int \ud(\mu, x)d\mu(x). \]
\end{defn}

The local dimension $d(\mu,x)$ captures the power-law scaling exponent of $\mu(B(x,r))$ as $r\to 0$, giving coarser information than the pointwise $s$-densities $\Theta^s(\mu,x) = \lim \limits_{r\to 0} \mu(B(x,r))/r^s$ studied in geometric measure theory (see~\cite[Chapter~6]{mattila}). More precisely, the local dimension retains only the critical exponent $s$ at which the $s$-density transitions from infinity to zero (see~\cite[Sections~1.9, 2.1]{BSSBook}), while the $s$-density at this critical exponent carries finer information about the multiplicative constant. The local dimensions are also closely related to the Hausdorff dimension $\hdim$ (a classical notion of dimension in geometric measure theory, see e.g. \cite[Section 4.3]{mattila}): for any Borel set $A$ with $\mu(\R^n \setminus A) = 0$, $\hdim(A) \geq \underset{x \sim \mu}{\esssup}\ \ld(\mu, x)$, and the Hausdorff dimension of the measure, defined as $\hdim(\mu) := \inf\{\hdim(A) : \mu(\R^n \setminus A) = 0\}$, satisfies $\hdim(\mu) = \underset{x \sim \mu}{\esssup}\ \ld(\mu, x)$, see \cite[Section~1.9]{BSSBook}. We therefore have $\laldim(\mu) \leq \hdim (\mu).$

Given a random variable $X$ taking values in $\R^n$, we will denote by $\laldim (X)$ and $\ualdim (X)$ the average local dimensions of the distribution of $X$ on $\R^n$, i.e. $\laldim (X) := \laldim (\mu_X)$ with $\mu_X$ defined by $\mu_X(A) = \PP(X \in A)$, where $X$ is a random vector on a probability space $(\Om, \Fk, \PP)$. We will use the same convention for all other notions of dimension that appear throughout the paper, e.g. $\lid(X) := \lid(\mu_X)$ for the information dimension defined below.

A useful basic fact (following e.g. from \cite[Theorem 1.9.5.(ii)]{BSSBook}) is that for a finite Borel measure on $\R^n$
\begin{equation}\label{eq: loc dim leq n}
	0 \leq \ld(\mu, x) \leq \ud(\mu,x) \leq n \text{ for } \mu\text{-a.e. } x.
\end{equation}
Consequently if $\mu$ is a probability measure, then
\begin{equation}\label{eq: aldim leq n}
	0 \leq \laldim(\mu) \leq \ualdim(\mu) \leq n
\end{equation}

\begin{defn} Let $\mu$ be a probability measure on $\R^n$. We say that $\mu$ is \textbf{local dimension regular}, if the local dimension of $\mu$ exists at $\mu$-a.e. $x \in \R^n$.
	Then, we define the  \textbf{average local dimension} of $\mu$ as
	\[ \aldim(\mu) = \int d(\mu, x) d\mu(x). \]
\end{defn}
Note that $\mu$ is local dimension regular if and only if $\laldim(\mu) = \ualdim(\mu)$ and then $\aldim(\mu)$ equals their common value.

\subsection{Information dimensions}
\begin{defn}\label{def: id}
	For a Borel probability measure $\mu$ on $\R^n$ the \textbf{lower} and \textbf{upper information dimensions} of $\mu$ are
	\[ \lid(\mu) = \liminf \limits_{r \to 0} \int \limits_{\supp(\mu)} \frac{\log \mu(B_2^n(x,r))}{\log r} d\mu(x) \text{ and } \uid(\mu) = \limsup \limits_{r \to 0} \int \limits_{\supp(\mu)} \frac{\log \mu(B_2^n(x,r))}{\log r} d\mu(x).\]
	If $\lid(\mu) = \uid(\mu)$, then we denote their common value as $\idim (\mu)$ and call it the \textbf{information dimension} of $\mu$.
\end{defn}

\begin{rem} Information dimensions of a non-compactly supported measure $\mu$ may be infinite if $\int \log \mu(B_2^n(x,r)) d\mu(x)$ is infinite for some $r > 0$. If however $\uid(\mu) < \infty$, then automatically $0 \leq \lid(\mu) \leq \uid(\mu) \leq n$. This will be so if $\mu$ has finite variance (in fact $\int \|x\|_2^\eps d\mu(x) < \infty$ for some $\eps>0$ suffices), see \cite[Proposition 1]{WV10} for details. Moreover, information dimensions can be alternatively defined as
	\begin{equation}\label{eq:id_equiv_def} \lid(\mu) = \liminf \limits_{\eps \to 0} \frac{1}{\log \eps} \sum \limits_{C \in \mathcal{C}_{\eps}}\mu(C)\log\mu(C) \text{ and } \uid (\mu) = \limsup \limits_{\eps \to 0} \frac{1}{\log \eps} \sum \limits_{C \in \mathcal{C}_{\eps}}\mu(C)\log\mu(C)
	\end{equation}
	where $\mathcal{C}_{\eps}$ is the partition of $\R^n$ into cubes with side length $\eps$ and vertices on the lattice $(\eps\Z)^n$, see e.g. \cite[Proposition 4]{WV10}. Moreover, it suffices to take (upper and lower) limits along the sequence $\eps_k = 1/k$ or $\eps_k = 2^{-k}$.
\end{rem}

By~\eqref{eq:id_equiv_def}, writing $H_\eps(\mu) := -\sum \limits_{C \in \mC_\eps}\mu(C)\log\mu(C)$ for the Shannon entropy of the $\eps$-quantization of $\mu$, the information dimension reads $\idim(\mu) = \lim \limits_{\eps \to 0} H_\eps(\mu) / (-\log \eps)$.
Thus $\idim(\mu)$ measures the growth rate of Shannon entropy under fine quantization: it is, in this sense, a Shannon entropy-based notion of dimension of a measure.

\noindent The following Lemma, giving a basic relationship between $\adim(\mu)$ and $\idim(\mu)$,  is proven in Appendix \ref{sec: lem aldim to idim proof}. 
\begin{lemma}\label{lem: aldim to idim}
	Let $\mu$ be a probability measure on $\R^n$ with finite variance. Then
	\begin{equation}\label{eq: aldim id comp} \laldim(\mu) \leq \lid (\mu)\leq \uid (\mu) \leq \ualdim(\mu) \leq n. \end{equation}
	Moreover, if $\mu$ is local dimension regular then $\aldim(\mu) = \idim(\mu)$ (in particular, both quantities exist).
\end{lemma}

\begin{example}\label{ex: non dyn}
	Lemma \ref{lem: aldim to idim} immediately gives a number of examples where $\aldim(\mu)$ is easy to compute. For instance, if $\mu$ is an absolutely continuous measure on a smooth $d$-dimensional submanifold in $\R^n$, then $\aldim(\mu) = d$, and if $\mu = (1-p)\mu_d + p\mu_c$, where $\mu_d$ is a discrete measure (on countably many atoms) and $\mu_c$ is an absolutely continuous measure in $\R^n$ (i.e. $\mu$ has a mixed distribution), then $\aldim(\mu) = pn$, see e.g. \cite{renyi1959dimension}. Moreover, measures with dynamical symmetries often are local dimension regular, e.g. invariant hyperbolic measures for $C^{1+\alpha}$ diffeomorphisms of Riemannian manifolds \cite{BPS99} or self-affine \cite{FengDimensionSelfAffine} and self-conformal measures \cite{FengHu09}. On the other hand, it is not difficult to construct measures with all inequalities in \eqref{eq: aldim id comp} being strict, see e.g. \cite[Section 3]{FanLauRao02}.
\end{example}

\subsection{Mean information dimension and information dimension rate}\label{subsec: mid}
Let us now consider several complexity measures of stochastic processes which are based on dimension notions from information theory and geometric measure theory. Through the paper, all stochastic processes are assumed to be $\R$-valued. Given a stochastic process $\mathbf{X} = (X_1 ,X_2, \ldots)$ we will use the notation $X_k^n := (X_k, \ldots, X_n)$ for $k,n \in \N \cup \{ \infty \}$ and a shorthand $X^n = X_1^n$. We will denote by $(\Om, \Fk, \PP)$ the underlying probability space. For $k \geq 1$ let $[X^n]_k := \frac{\lfloor k X^n \rfloor}{k}$ be the quantization of $X^n$ in scale $1/k$ (this is a random variable taking values in $(\frac{1}{k}\mathbb{Z})^n$). Let $H([X^n]_k)$ denote the Shannon entropy of $[X^n]_k$. Let $\mathbf{X}$ be stationary and such that $H([X^1]_1)<\infty$. The upper mean information dimension was defined in \eqref{eq: umid def}. In terms of the information dimensions, we can equivalently define the \textbf{upper and lower mean information dimensions} of a stationary stochastic process as
\[ \umid(\mathbf{X}) = \lim \limits_{n \to \infty} \frac{\uid(X^n)}{n}\ \text{ and }\ \lmid(\mathbf{X}) = \liminf \limits_{n \to \infty} \frac{\lid(X^n)}{n}. \]
The \textbf{upper and lower information dimension rates} of $\mathbf{X}$ are defined as
\[ \uidimr(\mathbf{X}) = \limsup \limits_{k \to \infty} \lim \limits_{n \to \infty} \frac{H([X^n]_k)}{n \log k}\ \text{ and }\ \lidimr(\mathbf{X}) = \liminf \limits_{k \to \infty} \lim \limits_{n \to \infty} \frac{H([X^n]_k)}{n \log k}. \]

In both definitions, whenever the (double) limit exist, we refer to it as the \textbf{information dimension rate}, denoted $\idimr(\mathbf{X})$, and the \textbf{mean information dimension}, denoted $\midim(\mathbf{X})$, respectively (in other words, $\idimr(\mathbf{X})$ exists if $\lidimr(\mathbf{X}) = \uidimr(\mathbf{X})$ and equals their common value, and likewise for $\midim(\mathbf{X})$). Note, comparing the above defintions with \eqref{eq: umid def}, that the definitions of $\midim(\cdot)$ and $\idimr(\cdot)$ differ only by the order of limits with respect to $n$ (the "time") and $k$ (corresponding to the "scale" $1/k$). The information dimension rate was introduced by Geiger and Koch \cite{geiger2019information}\footnote{see also \cite{GS21VarPrin} for a definition valid for general dynamical systems and \cite{YCZ25} for a panorama of related concepts.} and the mean information dimension by Jalali and Poor \cite{jalali2017universal} (the original definition is different, but it agrees with the above one by \cite[Lemma 3]{jalali2017universal}); note that we use different notation than in those papers. Geiger and Koch proved that the information dimension rate coincides with the rate-distortion dimension as defined by Rezagah et at \cite{rezagah2017compression} and inequalities
\begin{equation}\label{lem: idimr mid ineq}
	\uidimr(\mathbf{X}) \leq \umid(\mathbf{X}) \leq 1\ \text{ and }\ \lidimr(\mathbf{X}) \leq \lmid(\mathbf{X}) \leq 1
\end{equation}
hold \cite[Theorem 14]{geiger2019information}. As $\idimr(\mathbf{X}) = \lim \limits_{k\to\infty} h([X]_k)/\log k$ for $h([X]_k) = \lim \limits_{n \to \infty} \frac{H([X^n]_k)}{n \log k}$, we see that $\idimr(\mathbf{X})$ is the limit of entropy rates (see e.g. \cite[Section 4.2]{cover2012elements}) of the quantized process, normalized in a dimension-like fashion.

\subsection{Mean average local dimension}

Let us now define the mean average local dimension of a stochastic process.

\begin{defn}\label{def: mlaldim}
	Let $\mathbf{X} = (X_1, X_2, \ldots)$ be a stochastic process. Its \textbf{upper and lower mean average local dimensions} are defined as
	\[ \mlaldim(\mathbf{X}) = \liminf \limits_{n \to \infty} \frac{\laldim(X^n)}{n}\ \text{ and }\   \mualdim(\mathbf{X}) = \limsup \limits_{n \to \infty} \frac{ \ualdim(X^n)}{n}. \]
\end{defn}

In the following lemmas we compare tha mean average local dimensions with $\midim$ and $\idimr$. Let us begin with general sources.

\begin{lem}\label{lem: aldim mid idimr}
	Let $\mathbf{X} = (X_1, X_2, \ldots)$ be a stationary stochastic process with finite variance. Then
	\begin{enumerate}
		\item $\mlaldim(\mathbf{X}) \leq \lmid (\mathbf{X}) \leq 1$,
		\item $\uidimr (\mathbf{X}) \leq \umid(\mathbf{X}) \leq \mualdim(\mathbf{X}) \leq 1$.
	\end{enumerate}
\end{lem}

\begin{proof}
	This follows from Lemma \ref{lem: aldim to idim} and inequalities \eqref{lem: idimr mid ineq}.
\end{proof}

\begin{lem}\label{lem: aldim idim}
	Let $\mathbf{X} = (X_1, X_2, \ldots)$ be a stationary stochastic process with finite variance and assume that all finite-dimensional marginals of $\mathbf{X}$ are local dimension regular. Then
	\[ \mlaldim(\mathbf{X}) = \mualdim(\mathbf{X}) = \lim \limits_{n \to \infty}  \frac{\aldim(X^n)}{n} = \midim(\mathbf{X}). \]
\end{lem}

\begin{proof}
	Lemma \ref{lem: aldim to idim} gives $\laldim(X^n) = \ualdim(X^n) = \aldim(X^n) = \idim (X^n)$. The existence of the limit $\lim \limits_{n \to \infty} \frac{\idim (X^n)}{n}$ follows from the subadditivity of the sequence $n \mapsto \idim (X^n)$ (which in turn follows from the subadditivity of Shannon's entropy)
\end{proof}

The assumption of local dimension regularity of finite-dimensional distributions cannot be omitted in the above lemma. See Appendix \ref{sec: maldim noneq} for the details.

\begin{example}\label{ex: gaussian}
	Let $\mathbf{X} = (X_1, X_2, \ldots)$ be a stationary Gaussian process. Then $X^n$ has an absolutely continuous distribution on a $k$-dimensional linear subspace of $\R^n$, where $k = \rank(\Sigma_n)$ with $\Sigma_n$ being the covariance matrix of $X^n$. Therefore $X^n$ has local dimension regular finite-dimensional marginals, and hence Lemma \ref{lem: aldim idim} gives
	\[ \mlaldim(\mathbf{X}) = \mualdim(\mathbf{X}) = \midim(\mathbf{X}) = \lim \limits_{n \to \infty} \frac{\rank(\Sigma_n)}{n}. \]
	
	Combining this with \cite[Example 4]{geiger2019information} yields an existence of a stationary Gaussian process with
	\[  \idimr(\mathbf{X}) < \mlaldim(\mathbf{X}) = \mualdim(\mathbf{X}) = \midim(\mathbf{X}). \]
\end{example}

\begin{lem}\label{lem: psi star mlaldim}
	Let $\mathbf{X}$ be a stationary, $\psi^*$-mixing stochastic process. Then the limit defining $\mlaldim(\mathbf{X})$ exists, i.e. $ \mlaldim(\mathbf{X}) = \lim \limits_{n \to \infty}  \frac{\laldim(X^n)}{n}$. Moreover, $\umid(\mathbf{X}) = \uidimr(\mathbf{X})$ in this case.
\end{lem}

For the proof of the first statement see Appendix \ref{sec: lem psi star mlaldim proof}. The equality $\umid(\mathbf{X}) = \uidimr(\mathbf{X})$ in the above lemma is \cite[Corollary 15]{geiger2019information}.

For i.i.d processes with local dimension regular distributions, the mean average local dimension equals both $\aldim(X_1)$ and $\idim(X_1)$:

\begin{example}\label{ex: i.i.d.}
	Let $\mathbf{X}$ be an i.i.d process with local dimension regular $1$-dimensional distribution, then
	\begin{equation}\label{eq: ex i.i.d.}
		\mlaldim(\mathbf{X}) = \mualdim(\mathbf{X}) = \midim(\mathbf{X}) = \idimr(\mathbf{X}) = \aldim(X_1) = \idim (X_1)
	\end{equation}
	 This is a special case of Proposition~\ref{prop: markov psi star}. In general, \eqref{eq: ex i.i.d.} fails if the $1$-dimensional margin of the process is not local dimension regular. See Appendix \ref{sec: maldim noneq} for the details. 
	
	In particular, if $\bm{X}$ is i.i.d. with  $1$-dimensional margin $\mu$ of the form $\mu = p \mu_c + (1-p)\mu_d$, where $p \in [0,1]$, $\mu_c$ is absolutely continuous and $\mu_d$ is discrete (so $\bm{X}$ is a mixed discrete-continuous source), then combining \eqref{eq: ex i.i.d.} with Example \ref{ex: non dyn} yields
	\[ \aldim(\mathbf{X}) = \mualdim(\mathbf{X}) = \midim(\mathbf{X}) = \idimr(\mathbf{X}) = p. \]
\end{example}

\subsection{Energy and correlation dimension}\label{subsec: corr dim}

To prove Theorem~\ref{main theorem}, we first prove a similar result for general sources in terms of a new complexity measure of a stochastic process, which is inspired by the correlation dimension and related techniques from geometric measure theory, see e.g. \cite[Chapters 8-10]{mattila} or \cite[Chapter 3]{BP17}. For $s \geq 0$, the $\mathbf{s}$\textbf{-energy} of a finite Borel measure $\mu$ on $\R^n$ is
\[\mE_s (\mu) := \int \int \|x-y\|_2^{-s}d\mu(x)d\mu(y)\]
(recall that $\|\cdot\|_2$ stands for the Euclidean norm on $\R^n$).

\begin{defn}
	For a finite Borel measure $\mu$ on $\R^n$, its \textbf{correlation dimension} is defined as
	\[ \dim_\mathrm{cor}(\mu) = \sup\{ s \geq 0 : \Ek_s(\mu) < \infty\}.  \]
\end{defn}
It is easy to see that the set $\{ s \geq 0 : \Ek_s(\mu) < \infty\}$ is an interval. The correlation dimension defined as above is also called the \textit{lower correlation dimension} or the $L^2$\textit{-dimension}, see \cite[Sections 1.9.3 and 2.6]{BSSBook} for a more detailed discussion. A basic fact about the correlation dimension is
\begin{equation}\label{eq: dimcor dimh n ineq} 0 \leq \dim_\mathrm{cor}(\mu) \leq \underset{x \sim \mu}{\essinf}\ \ld(\mu, x) \leq \laldim(\mu) \leq n\ \text{ for every finite Borel measure } \mu \text{ on } \R^n,
\end{equation}
see e.g. \cite[Theorem 1.4]{FanLauRao02}. It is also easy to see that if $\mu$ has an atom, then $\dim_\mathrm{cor} \mu = 0$. 
We will use repeatedly the following formula (see e.g. \cite[p. 109]{mattila}), valid for a finite Borel measure $\mu$ on $\R^n$ and $0 < s < n$ and $x \in \R^n$
\begin{equation}\label{eq: energy integral}
	\int \|x-y\|_2^{-s}d\mu(y) = s\int \limits_{0}^\infty r^{-s-1}\mu(B(x,r))dr
\end{equation}

The correlation dimension relates to the local dimensions and the Hausdorff dimension discussed earlier. The main difference is that bounding it from below requires uniform control on the measure of balls (via $s$-energy integrals $\iint \|x-y\|^{-s}\,d\mu(x)\,d\mu(y)$) rather than a pointwise local dimension control; see~\cite[Section~8]{mattila}. For $\mu$ supported in a compact set $A$, one has the chain of inequalities
\[
\dim_{\mathrm{cor}}(\mu) \leq \hdim(\mu) \leq \hdim(A) \leq \underline{\dim}_B(A) \leq \overline{\dim}_B(A),
\]
where $\underline{\dim}_B(A)$ and $\overline{\dim}_B(A)$ denote, respectively, the upper and lower box-counting (Minkowski) dimensions of $A$; see~\cite[Sections~1.4, 1.9]{BSSBook}.

\subsection{Correlation dimension rate}
Similarly as with the information dimension and average local dimension above, we can define a complexity measure of a stochastic process based on the correlation dimension. The following defintion turns out to be well suited for our needs.
\begin{defn}\label{def: mdimcor}
	For each $n \geq 1$, let $\mu_n$ be a finite Borel measure on $\R^n$. The \textbf{correlation dimension rate} of the sequence $(\mu_n)_{n=1}^\infty$ is
	\[ \mdimcor ( (\mu_n)_{n=1}^\infty) := \sup \left\{ \theta \geq 0 : \limsup \limits_{n \to \infty} \frac{1}{n} \log \left( n^{\theta n/2} \mE_{\theta n}(\mu_n) \right) < \infty \right\}. \]
	For a stochastic process $\mathbf{X} = (X_1, X_2, \ldots)$ we define
	\[ \mdimcor(\mathbf{X}) := \mdimcor ( (\mu_{X^n})_{n=1}^\infty), \]
	where $\mu_{X^n}$ is the distribution of $X^n$ on $\R^n$.
\end{defn}

In terms of the asymptotic notation from Section \ref{subsec: conventions}, the definition of the correlation dimension rate can be equivalently written as follows:
\begin{equation}\label{eq: mdimcor lesssim e} \mdimcor ( (\mu_n)_{n=1}^\infty) = \sup \left\{ \theta \geq  0 : \mE_{\theta n}(\mu_n) \lesssim^e_{\theta} n^{-\theta n/2}  \right\}
\end{equation}
(Recall from Section~\ref{subsec: conventions} that the constant in $\lesssim^e_\theta$ is allowed depend on the fixed sequence $(\mu_n)$; this dependence is suppressed in the notation).

An immediate consequence of \eqref{eq: dimcor dimh n ineq} is the following inequality, valid for an arbitrary stochastic process $\mathbf{X}$
\begin{equation}\label{eq: mdimcor geq 1} \mdimcor ( \mathbf{X}) \leq \liminf \limits_{n \to \infty} \frac{\dim_\mathrm{cor}(X^n)}{n} \leq \mlaldim(\mathbf{X}) \leq 1.
\end{equation}

The following example shows how the normalizing term $n^{\theta n /2}$ appears naturally and provides a direct computation of the mean correlation dimension rate.  . It proves that $\mdimcor(\mathbf{X}) = 1$ for $\mathbf{X}$ being an i.i.d. process with a uniform distribution on an interval as the one-dimensional margin.

\begin{example}\label{ex: mdimcor leb}
	Let $\mu_n = \Leb_n|_{[0,1]^n}$. Then $\mdimcor ( (\mu_n)_{n=1}^\infty) = 1$. To prove this, see first that by \eqref{eq: vol unit ball} and \eqref{eq: gamma bounds}, for every $x \in \R^n$,
	\[ \mu_n(B_2^{n}(x,r)) \leq \alpha(n)r^n \lesssim^e \frac{r^n}{\Gamma(\frac{n}{2} + 1)}\lesssim^e n^{-n/2} r^n. \]
	Therefore for $0 < \theta < 1$ by \eqref{eq: energy integral}
	\[\begin{split}  \mE_{\theta n}(\mu_n) & = \theta n \int \int \limits_{0}^\infty r^{-\theta n - 1} \mu_n(B_2^n(x,r))dr d\mu_n(x) \lesssim^e  n^{-n/2}\int \limits_{0}^{\sqrt{n}}r^{(1-\theta)n - 1}dr + \int \limits_{\sqrt{n}}^\infty r^{-\theta n - 1} dr \\
		&	= n^{-n/2} \frac{1}{(1-\theta)n}n^{(1-\theta)n/2} + \frac{n^{-\theta n /2}}{\theta n} \lesssim^e_{\theta} n^{-\theta n / 2}.
	\end{split}\]
	Therefore $\mdimcor ( (\mu_n)_{n=1}^\infty) \geq 1$ by \eqref{eq: mdimcor lesssim e}. The upper bound follows from \eqref{eq: mdimcor geq 1}.
\end{example}

\begin{rem}
The above calculation also illustrates a useful principle for bounding energy integrals: we split the integration variable $r$ at $\sqrt{n}$, separating the contributions from $r \leq \sqrt{n}$ and $r > \sqrt{n}$. The tail integral $\int_{\sqrt{n}}^\infty r^{-\theta n - 1}\,dr$ already exhibits the desired exponential decay regardless of $\mu_n(B_2^n(x,r))$, so the non-trivial task is to bound the integral over $[0, \sqrt{n}]$. We use splittings of this type in several proofs throughout the paper, including the proofs of Theorem~\ref{thm: mdimcor detail}, Proposition~\ref{prop: restriction}, and Lemma~\ref{lem: energy comparsion}.
\end{rem}

\subsection{Relations between entropy and dimensions}\label{subsec: discussion}

The names \textit{mean information dimension} and \textit{mean average local dimension} are adopted by analogy with the notions of mean dimensions studied in topological dynamics and ergodic theory, introduced by Gromov~\cite{G} and further developed by Lindenstrauss and Weiss~\cite{LW} and others (e.g., \cite{LT12, gutman2020embedding,gutman2019application}); see also~\cite{YCZ25} for a discussion of entropy-based notions of measure-theoretic mean dimension. These invariants have been shown to be well suited for the study of infinite-entropy systems.

In general there are no direct connections between dimensions and entropy for arbitrary measures or stochastic processes. There are, however, deep connections in the class of smooth dynamical systems - see the celebrated theory of Ledrappier and Young~\cite{LY85I,LY85II} relating dimensions, entropy, and Lyapunov exponents, as well as its extension to self-similar and self-affine measures~\cite{BK17}.

\section{A converse for general sources in terms of the correlation dimension rate}\label{sec: mdimcor main}

For technical reasons, it will be convenient for us to work with \textit{subprobability} measures, i.e. measures on a measure space $X$ satisfying $\mu(X) \leq 1$.

Let $G_{n,m}$ denote the standard Gaussian measure on $\R^{m \times n}$ with $m \leq n$, i.e.\ we identify $\R^{m \times n}$ with $m \times n$ matrices and $G_{n,m}$ is the distribution of a random matrix $A = [a_{i j}]$, where $a_{i j}$ are i.i.d.\ with standard Gaussian distribution $N(0,1)$. Note that with this notation the measure $\nu$ appearing in Theorem~\ref{main theorem} is given as $\nu = \bigotimes \limits_{n=1}^\infty G_{n,m_n}$

\begin{thm}[Main Technical Theorem]\label{thm: mdimcor detail}
	Fix $M \geq 1$. For each $n \geq 1$, let $\mu_n$ be a subprobability measure on $B_2^n(0, \sqrt{n}M)$. Let $m_n  \in \N$ be a sequence such that $\liminf \limits_{n \to \infty} \frac{m_n}{n} < \mdimcor ( (\mu_n)_{n=1}^\infty) $. Let  $F_n : \R^{m_n} \times \R^{m_n \times n} \to \R^n$ be a sequence of Borel maps. Then there exists $\delta_0$ such that for every $0 < \delta \leq \delta_0$
	\[ \liminf \limits_{n \to \infty}\ \mu_n \otimes G_{n,m_n} \left( \left\{ (x, A) \in \R^n \times \R^{m_n \times n} : \frac{1}{\sqrt{n}}\|x - F_n(Ax, A)\|_2 \leq \delta  \right\} \right) = 0.\]
\end{thm}

\begin{rem}
	Note that the threshold $\liminf \limits_{n \to \infty} \frac{m_n}{n} < \mdimcor( \mathbf{X})$ cannot be optimal in general. For instance, the mixed discrete-continuous source from Example \ref{ex: sparse} satisfies $\mdimcor(\mathbf{X}) = 0$ if $p<1$, as then every finite-dimensional distribution $\mu_n$ of the process has an atom, hence $\dim_{\mathrm{corr}}(\mu_n) = 0$. On the other hand, as discussed in Example \ref{ex: sparse}, it follows from Theorem~\ref{main theorem} that $\frac{1}{\sqrt{n}} \|X^n - \hat{X}^n\|_2 \text{ does not converge to } 0 \text{ in probability}$ already if $\liminf \limits_{n \to \infty} \frac{m_n}{n} < p$.
\end{rem}

\begin{proof}[{\bf Proof of Theorem \ref{thm: mdimcor detail}}]
	Fix $\theta, R >0$ such that $\liminf \limits_{n \to \infty} \frac{m_n}{n} < R < \theta < \mdimcor((\mu_n)_{n=1}^{\infty}) \leq 1$  (recall \eqref{eq: mdimcor geq 1}) and let $n_k \nearrow \infty$ be a sequence such that $\lim \limits_{k \to \infty} \frac{m_{n_k}}{n_k}$ exists and $m_{n_k} \leq Rn_k$ for all $k$. In particular by Lemma \ref{lem: energy comparsion}
	\begin{equation}\label{eq: theta' energy}
		\Ek_{\theta' n} \lesssim^e_{\theta} n^{-\theta' n / 2} \text{ for every } 0 < \theta' \leq \theta.
	\end{equation}
	
	We shall prove that there exists $\delta_0$ such that for every $0 < \delta \leq \delta_0$
	\begin{equation}\label{eq: prob subseq}  \lim \limits_{k \to \infty}\ \mu_{n_k} \otimes G_{n_k,m_{n_k}} \left( \left\{ (x, A) \in  \R^{n_k} \times \R^{m{_{n_k}} \times n_k}  : \frac{1}{\sqrt{n_k}}\|x - F_{n_k}(Ax, A)\|_2 \leq \delta  \right\} \right) = 0.
	\end{equation}
	For short, let us write $n = n_k$ and $m=m_{n_k}$.  Let $K$ be the constant from Lemma \ref{lem: gauss norm} and set
	\[ Q_n = \{ A \in \R^{m_n \times n} : \|A\| \leq K\sqrt{n}\}\]
	and for $A \in \R^{m_n \times n}$
	\[T_n(A) = \{ z \in \R^m :  \underset{0 < r \leq 1}{\forall}\ \mu_n(A^{-1}(B_2^m(z,r))) > 2^{-n}(10KMn)^{-m}r^m  \} \]
	(recall that $M$ is fixed in the statement of the theorem to be proved). By Lemma \ref{lem: gauss norm} we have
	\[G_{n,m}(Q_n^c) \leq 2e^{-n}\]
	and by Lemma \ref{lem: density bound} applied with $D = 2^{-n}(10KMn)^{-m}$ we have for $A \in Q_n$ and $n$ large enough
	\[ \mu_n(A^{-1}(T_n(A))^c) \leq 2^{-n}, \]
	as $5\|A\|\sqrt{n}M+1 \leq 10KnM$ for $n$ large enough and $A \in Q_n$.
	Therefore, setting
	\[
	\begin{split}E_n & = \left\{ (x, A) \in \R^n \times \R^{m \times n} : \|A\| \leq K\sqrt{n}\ \text{ and }\ \underset{0 < r \leq 1}{\forall}\ A\mu(B_2^m(Ax,r)) > 2^{-n}(10KMn)^{-m}r^m \right\} \\
		& = \bigcup \limits_{A \in Q_n} A^{-1}(T_n(A)) \times \{A\}
	\end{split}
	\]
	we have by Fubini's theorem that
	\begin{align*}
		\mu_n \otimes G_{n,m} (E_n^c) & \leq G_{n,m}(Q_n^c) \cdot \mu_n(\R^n) + \int_{Q_n} \mu_n\bigl(A^{-1}(T_n(A))^c\bigr)\, dG_{n,m}(A) \\
		& \leq 2e^{-n} \cdot 1 + 2^{-n} \cdot 1 = 2e^{-n} + 2^{-n} \to 0 \text{ as } n \to \infty.
	\end{align*}
	Consequently, it suffices to prove that there exists $\delta_0$ such that for every $0 < \delta \leq \delta_0$
	\begin{equation}\label{eq: En prob} \lim \limits_{k \to \infty}\ \mu_{n_k} \otimes G_{n_k,m_{n_k}} \left( \left\{ (x, A) \in E_{n_k} : \frac{1}{\sqrt{n_k}}\|x - F_{n_k}(Ax, A)\|_2 \leq \delta  \right\} \right) = 0.
	\end{equation}
	With the use of the disintegration \eqref{eq:cond meas decomp} of $\mu_n$ into conditional measures $\mu_{n,A,z}$ (with respect to the map $\phi = A : \R^n \to \R^m$; in this case the fiber $A^{-1}z$ is an affine subspace of $\R^n$; see Appendix~\ref{sec: cond} for the discussion of the properties of the conditional disintegration that we use) we can write as follows for every $s > 0$
	\[
	\begin{split}
		\mu_n \otimes G_{n,m} & \left( \left\{ (x, A) \in E_n : \frac{1}{\sqrt{n}}\|x - F_n(Ax, A)\|_2 \leq \delta  \right\} \right) \\
		& = \int \limits_{Q_n} \int \limits_{T_n(A)} \mu_{n,A,z}\left( \left\{ x \in A^{-1}(T_n(A)) : \frac{1}{\sqrt{n}}\|x - F_n(Ax, A)\|_2 \leq \delta  \right\}  \right) dA\mu_n(z)  dG_{n,m}(A) \\
		& \overset{\eqref{eq:cond meas inverse}}{=}  \int \limits_{Q_n} \int \limits_{T_n(A)} \mu_{n,A,z}\left( \left\{ x \in \R^n : \frac{1}{\sqrt{n}}\|x - F_n(z, A)\|_2 \leq \delta  \right\}  \right) dA\mu_n(z)  dG_{n,m}(A) \\
		& = \int \limits_{Q_n} \int \limits_{T_n(A)} \mu_{n,A,z}\left(B_2^n(F_n(z, A), \sqrt{n}\delta) \right) dA\mu_n(z)  dG_{n,m}(A) \\
		& \overset{\text{Lem. } \ref{lem: ball to energy}}{\leq} 2^{s/2} n^{s/4}\delta^{s/2}\int \limits_{Q_n} \int \limits_{T_n(A)} \mE_s(\mu_{n,A,z})^\frac{1}{2}dA\mu_n(z)  dG_{n,m}(A).
	\end{split}
	\]
	Applying Jensen's inequality and recalling that $\mu_n(\R^n) \leq 1$ gives for every $s>0$
	\begin{equation}\label{eq: measure s bound}
		\begin{split}
			\mu_n \otimes G_{n,m} & \left( \left\{ (x, A) \in E_n : \frac{1}{\sqrt{n}}\|x - F_n(Ax, A)\|_2 \leq \delta  \right\} \right) \\
			& \leq  2^{s/2} n^{s/4}\delta^{s/2}  \left(\int \limits_{Q_n} \int \limits_{T_n(A)} \mE_s(\mu_{n,A,z})dA\mu_n(z)  dG_{n,m}(A) \right)^\frac{1}{2}.
		\end{split}
	\end{equation}
	Let us now bound the above integral. We have by \eqref{eq:cond measure limit} and the lower semi-continuity of the function $x \mapsto \|x-y\|_2^{-s}$ on $\R^n$
	\[
	\begin{split}
		\int \limits_{Q_n} & \int \limits_{T_n(A)}  \mE_s(\mu_{n,A,z})dA\mu_n(z)  dG_{n,m}(A) = \int \limits_{Q_n} \int \limits_{T_n(A)} \int \limits_{\R^n} \int \limits_{\R^n} \|x-y\|_2^{-s} d\mu_{n,A,z}(x)d\mu_{n,A,z}(y)dA\mu_n(z)  dG_{n,m}(A) \\
		& \overset{\eqref{eq:muG liminf}}{\leq} \int \limits_{Q_n} \int \limits_{T_n(A)}  \int \limits_{\R^n} \liminf \limits_{r \to 0} \int \limits_{\R^n} \frac{\|x-y\|_2^{-s}\mathds{1}_{B(z,r)}(Ax)}{\mu_n(A^{-1}(B(z,r)))} d\mu_n(x)d\mu_{n,A,z}(y)dA\mu_n(z)  dG_{n,m}(A) \\
		& \overset{\text{Fatou's lem.}}{\leq} \liminf \limits_{r \to 0} \int \limits_{Q_n} \int \limits_{T_n(A)}  \int \limits_{\R^n}  \int \limits_{\R^n} \frac{\|x-y\|_2^{-s}\mathds{1}_{B(z,r)}(Ax)}{\mu_n(A^{-1}(B(z,r)))} d\mu_n(x)d\mu_{n,A,z}(y)dA\mu_n(z)  dG_{n,m}(A) \\
		& \overset{\eqref{eq:cond meas inverse}}{= } \liminf \limits_{r \to 0} \int \limits_{Q_n} \int \limits_{T_n(A)}  \int \limits_{\R^n}  \int \limits_{\R^n} \frac{\|x-y\|_2^{-s}\mathds{1}_{B(Ay,r)}(Ax)}{\mu_n(A^{-1}(B(z,r)))} d\mu_n(x)d\mu_{n,A,z}(y)dA\mu_n(z)  dG_{n,m}(A) \\
		& \overset{\text{def. of } T_n(A)}{\leq } 2^{n}(10KMn)^{m} \liminf \limits_{r \to 0} r^{-m} \int \limits_{Q_n} \int \limits_{T_n(A)}  \int \limits_{\R^n}  \int \limits_{\R^n} \|x-y\|_2^{-s}\mathds{1}_{B(Ay,r)}(Ax) d\mu_n(x)d\mu_{n,A,z}(y)dA\mu_n(z)  dG_{n,m}(A) \\
		& \overset{\eqref{eq:cond meas decomp} \text{ and } m \leq Rn}{ \lesssim^e_{M,R}} n^m \liminf \limits_{r \to 0} r^{-m} \int \limits_{Q_n} \int \limits_{\R^n}  \int \limits_{\R^n} \|x-y\|_2^{-s}\mathds{1}_{\{ \|Ax - Ay\|_2 \leq r \}} d\mu_n(x)d\mu_n(y) dG_{n,m}(A) \\
		& \overset{\text{Fubini's thm.}}{ \lesssim^e_{M,R}} n^m \liminf \limits_{r \to 0} r^{-m} \int \limits_{\R^n}  \int \limits_{\R^n} \|x-y\|_2^{-s}G_{n,m}\left(\{ A \in \R^{m_n \times n} :  \|Ax - Ay\|_2 \leq r \}\right) d\mu_n(x)d\mu_n(y).
	\end{split}
	\]
	For $r > 0$ we bound the last integral as follows, applying in the second inequality below Lemma \ref{lem: gauss trans} with $u = x - y$ and $\eps = \frac{r}{\sqrt{m}\|x-y\|_2}$
	\[
	\begin{split}
		r^{-m}\int \limits_{\R^n}    \int \limits_{\R^n} &  \|x-y\|_2^{-s}G_n \left(\{ A \in \R^{m \times n} :  \|Ax - Ay\|_2 \leq r \}\right) d\mu_n(x)d\mu_n(y) \\
		& \leq r^{-m} \iint \limits_{\{\sqrt{m}\|x-y\|_2 \leq r\}}  \|x-y\|_2^{-s}d\mu_n(x)d\mu_n(y)  \\
		& \qquad + r^{-m} \iint \limits_{\{\sqrt{m}\|x-y\|_2 >r\}} \|x-y\|_2^{-s} G_{n,m}\left(\{ A \in \R^{m \times n} :  \|Ax - Ay\|_2 \leq r \}\right) d\mu_n(x)d\mu_n(y) \\
		& \overset{\text{Lem. } \ref{lem: gauss trans}}{\leq} m^{-m/2}  \iint \limits_{\{\sqrt{m}\|x-y\|_2 \leq r\}}  \|x-y\|_2^{-(s+m)}d\mu_n(x)d\mu_n(y)  \\
		& \qquad + e^m m^{-m/2} \iint \limits_{\{\sqrt{m}\|x-y\|_2 >r\}} \|x-y\|_2^{-(s+m)} d\mu_n(x)d\mu_n(y) \\
		& \leq e^m m^{-m/2} \mE_{s+m}(\mu_n).
	\end{split}
	\]
	Combining the last two calculations gives
	\begin{equation}\label{eq: energy trans} \int \limits_{Q_n}  \int \limits_{T_n(A)}  \mE_s(\mu_{n,A,z})dA\mu_n(z)  dG_{n,m}(A) \lesssim^e_{M,R} n^m m^{-m/2} \Ek_{s+m}(\mu_n).
	\end{equation}
	Apply now \eqref{eq: measure s bound} and \eqref{eq: energy trans} with $s = s(n) = (\theta - R) n$ (so that $s + m \leq \theta n$) and \eqref{eq: theta' energy} to obtain
	\begin{equation}\label{eq: measure theta bound}
		\begin{split}
			\mu_n \otimes G_{n,m} & \left( \left\{ (x, A) \in E_n : \frac{1}{\sqrt{n}}\|x - F_n(Ax, A)\|_2 \leq \delta  \right\} \right) \\
			& \lesssim^e_{M, R, \theta} \delta^{s / 2} n^{s / 4} n^{m/2} m^{-m/4} \Ek_{s+m}(\mu_n)^{1/2} \\
			& \overset{\text{Lem. }\ref{lem: energy comparsion}}{\lesssim^e_{M, R, \theta}} \delta^{s / 2} n^{s / 4} n^{m/2} m^{-m/4} n^{-(s+m)/4} \\
			& \lesssim^e_{M, R, \theta} \delta^{(\theta - R) n / 2} (m/n)^{-m/4}.
		\end{split}
	\end{equation}
	Let $R' = \lim \limits_{k \to \infty} \frac{m_{n_k}}{n_k}$ (recall that we have chosen subsequence $n_k$ so that the limit exists). To finish the proof it suffices to prove
	\begin{equation}\label{eq: mn bound} (m/n)^{-m/4} \lesssim^e_{R'} 1,
	\end{equation}
	as then
	\[ \delta^{(\theta - R) n / 2} (m/n)^{-m/4} \lesssim^e_{M,R,R',\theta} \delta^{(\theta - R) n / 2},\]
	so by \eqref{eq: measure theta bound} there exists $C = C(M,R,R',\theta) \geq 0$ such that
	\[ \mu_n \otimes G_{n,m} \left( \left\{ (x, A) \in E_n : \frac{1}{\sqrt{n}}\|x - F_n(Ax, A)\|_2 \leq \delta  \right\} \right) \leq C^n \delta^{(\theta - R) n / 2}.\]
	Therefore, choosing $\delta_0 = \delta_0(M,R,R',\theta)$ such that $C\delta_0^{(\theta - R)/2} < 1$ implies that \eqref{eq: En prob} holds for every $0 < \delta \leq \delta_0$ and finishes the proof of Theorem \ref{thm: mdimcor detail}.
	To prove \eqref{eq: mn bound} we shall consider two cases. If $R' = 0$ , then
	\[ \frac{1}{n} \log \left((m/n)^{-m/4} \right) = \frac{-m}{4n} \log \frac{m}{n} \to 0 \text{ as } k \to \infty\]
	since $x \log x \to 0$ as $x \to 0$ and so \eqref{eq: mn bound} holds. 
	Otherwise $R' > 0$, so $m \geq R' n/2$ for all $k$ large enough, so for such $k$
	\[ (m/n)^{-m/4} \leq  (R'/2)^{-m/4} \leq (R'/2)^{-R' n / 8}\]
	and hence \eqref{eq: mn bound} holds in this case as well.
\end{proof}

\section{A converse for $\psi^*$-mixing stochastic process}\label{sec: laldim main}

\subsection{Converse in terms of mean average local dimension}

\begin{thm}\label{thm: laldim main}
	Let $\mathbf{X} = (X_1, X_2, \ldots)$ be a finite variance, stationary, $\psi^*$-mixing stochastic process. Consider a sequence $m_n  \in \N$ such that $\liminf \limits_{n \to \infty} \frac{m_n}{n} < \mlaldim(\mathbf{X})$. Let  $F_n : \R^{m_n} \times \R^{m_n \times n} \to \R^n$ be a sequence of Borel maps (where we identify $\R^{m_n \times n}$ with the space of linear maps $A : \R^n \to \R^{m_n}$). Let $A_n \in \R^{m_n \times n}$ be a sequence of random matrices with independent $N(0,1)$ entries, chosen independently of one another and of the process $\mathbf{X}$. Denote $\hat{X}^n = F_n(A_n X^n, A_n)$. Then
	\[ \frac{1}{\sqrt{n}} \|X^n - \hat{X}^n\|_2 \text{ does not converge to } 0 \text{ in probability}. \]
\end{thm}

The proof is given after Proposition~\ref{prop: restriction} below.

It remains an open problem whether this result can be improved to general stationary stochastic processes.
\begin{proof}[{\bf Proof of Theorem~\ref{main theorem}}]
	Theorem~\ref{main theorem} follows directly from Theorem \ref{thm: laldim main} and Lemma \ref{lem: aldim idim}.
\end{proof}

\subsection{$\psi^*$-mixing lemma}

We will use the $\psi^*$-mixing condition via the following lemma. We shall use the following  notation: for a vector $x=(x_1, \ldots, x_n) \in \R^n$ and $1 \leq i \leq j \leq n$ we set $x_i^j = (x_i, \ldots, x_j)$.

\begin{lemma}\label{lem: psi star on balls}
	Let $\mathbf{X} = (X_1, X_2, \ldots)$ be a stationary stochastic process on a probability space $(\Om, \Fk, \PP)$. Let $\mu_n$ denote the distribution of $X_1^n$ and let $g \in \N$ be such that $\psi^*(g) < \infty$.\footnote{We will use the $\psi^*$-mixing condition only through Lemma \ref{lem: psi star on balls} and hence a seemingly weaker condition would suffice: there exists $g \in \N$ such that $\psi^*(g) < \infty$. However by \cite[Theorem 1]{BradleyMixing83}, for mixing processes, this is equivalent to the $\psi^*$-mixing condition.} Then the following holds for every $i, k \in \N, r>0$ and $x \in \R^{i+g+k-1}$
	\[ \mu_{i+g+k-1}(B_2^{i+g+k-1}(x,r)) \leq \psi^*(g) \int \limits_{B_2^k(x_{i+g}^{i+g+k-1}, r)} \mu_i\left(B_2^i \left(x_1^i, \left( r^2 - \|x_{i+g}^{i+g+k-1} - z\|_2^2\right)^{1/2}\right)\right)d\mu_k(z). \]
\end{lemma}

\begin{proof}
	
	Note that while the lemma is stated for closed balls, it suffices to prove it for open balls (by the continuity of measure from above). Therefore, in the following proof we abuse the notation and let $B_2^n(x,r)$ denote the open $r$-ball in the Euclidean metric.
	
	It will be useful for us to consider the conditional disintegration of $\mu_{i+g+k-1}$ with respect to the projection map $\pi : \R^{i+g+k-1} \to \R^k,\ \pi(x_1, \ldots, x_{i+g+k-1}) = (x_{i+g}, \ldots, x_{i+g+k-1})$, as described in Section \ref{sec: cond} (in other words, we study conditional distribution of $X_1^{i+g+k-1}$ with respect to $X_{i+g}^{i+g+k-1}$). Note that by stationarity $\pi \mu_{i+g+k-1} = \mu_k$. Let $\mu_{\pi, z},\ z \in \R^k$ be the conditional distributions of $\mu_{i+g+k-1}$ with respect to $\pi$, so that by \eqref{eq:cond meas decomp} and \eqref{eq:cond meas inverse} for Borel $E \subset \R^{i+g+k-1}$
	\[ \mu_{i+g+k-1}(E) = \int \limits_{\R^k} \mu_{\pi, z}(E) d\mu_k(z) \]
	and $\mu_{\pi, z} \left( \left\{ x \in \R^{i+g+k-1} : x_{i+g}^{i+g+k-1} = z \right\} \right) = 1$ for $\mu_k$-a.e. $z \in \R^k$. We therefore have the following for $x = (x_1, \ldots, x_{n+g+k-1})$
	
	\begin{equation}\label{eq: conditional ball}
		\begin{split} \mu_{i+g+k-1} & (B_2^{i+g+k-1}(x,r))  \\
			& = \PP \left (\sum \limits_{j=1}^{i+g+k-1}|X_j - x_j|^2 < r^2 \right) \\
			& \leq \PP\left(\sum \limits_{j=1}^{i}|X_j - x_j|^2 + \sum \limits_{j=i+g}^{i+g+k-1} |X_j - x_j|^2 < r^2\right) \\
			& = \int \limits_{B_2^k(x_{i+g}^{i+g+k-1}, r)} \mu_{\pi, z}\left( \left\{ y \in \R^{i+g+k-1} : \sum \limits_{j=1}^{i}|y_i - x_i|^2 + \sum \limits_{j=i+g}^{i+g+k-1} |y_i - x_i|^2 < r^2 \right\} \right)d\mu_k(z) \\
			& = \int \limits_{B_2^k(x_{i+g}^{i+g+k-1}, r)} \mu_{\pi, z}\left( \left\{ y \in \R^{i+g+k-1} :  \sum \limits_{j=1}^{i}|y_i - x_i|^2 < r^2 - \|x_{i+g}^{i+g+k-1} - z\|_2^2 \right\} \right)d\mu_k(z). \\
		\end{split}
	\end{equation}
	
	By \eqref{eq:cond measure limit}, definition of $\psi^*(g)$ and the Portmanteau theorem (see e.g. \cite[Corollary 8.2.10]{BogachevMeasureTheory}; this is the reason for which we want to work with open balls), for $\mu_k$-a.e. $z \in \R^k$
	\[
	\begin{split}
		\mu_{\pi, z} & \left( \left\{ y \in \R^{i+g+k-1} : \sum \limits_{j=1}^{i}|y_i - x_i|^2 < r^2 - \|x_{i+g}^{i+g+k-1} - z\|_2^2  \right\} \right) \\
		&  \leq \liminf \limits_{\rho \to 0} \frac{\mu_{i+g+k-1}\left( \left\{ y \in \R^{i+g+k-1} :  \|y_{i+g}^{i+g+k-1} - z\|_2 \leq \rho \text{ and }\sum \limits_{j=1}^{i}|y_i - x_i|^2 < r^2 - \|x_{i+g}^{i+g+k-1} - z\|_2^2  \right\} \right)}{\mu_{i+g+k-1}\left( \left\{ y \in \R^{i+g+k-1} :  \|y_{i+g}^{i+g+k-1} - z\|_2 \leq \rho \right\} \right)} \\
		& \leq \psi^*(g) \mu_{i+g+k-1}\left( \left\{ y \in \R^{i+g+k-1} : \sum \limits_{j=1}^{i}|y_i - x_i|^2 < r^2 - \|x_{i+g}^{i+g+k-1} - z\|_2^2  \right\} \right) \\
		& = \psi^*(g) \mu_{i}\left( \left\{ y \in \R^{i} : \sum \limits_{j=1}^{i}|y_i - x_i|^2 < r^2 - \|x_{i+g}^{i+g+k-1} - z\|_2^2  \right\} \right) \\
		& = \psi^*(g) \mu_{i}\left( \left\{ y \in \R^{i} : \|y - x_{1}^{i}\|_2^2 < r^2 - \|x_{i+g}^{i+g+k-1} - z\|_2^2 \right\} \right) \\
		& = \psi^*(g) \mu_{i}\left( B_2^i \left( x_{1}^{i}, \left( r^2 - \|x_{i+g}^{i+g+k-1} - z\|_2^2 \right)^{1/2}  \right) \right).
	\end{split} 
	\]
	Combining this with \eqref{eq: conditional ball} finishes the proof.
\end{proof}

\subsection{Relating correlation dimension rate and mean local average dimension}

The main step for proving Theorem \ref{thm: laldim main} is the following proposition (note that we do not assume here the finite-dimensional marginals of the process to be local dimension regular).

\begin{prop}\label{prop: restriction}
	Let $\mathbf{X} = (X_1, X_2, \ldots)$ be a stationary, $\psi^*$-mixing stochastic process taking values in $\R$. Let $\mu_n$ be the distribution of $X_1^n$ an assume that $\Var(X_1) < \infty$ and $\E X_1 = 0$. Then for every $0< \eta < 1$ and $M \geq 1$ there exists a sequence of Borel sets $E_n \subset \R^n$ such that
	\begin{enumerate}
		\item\label{it: En ball} $E_n \subset B_2^n(0, \sqrt{n}M)$,
		\item\label{it: En measure bound} $\liminf \limits_{n \to \infty } \mu_n(E_n) \geq 1 - \frac{Var(X_1)}{M^2}$,
		\item\label{it: En mdimcor} $\mdimcor( (\mu_n|_{E_n})_{n=1}^\infty) \geq \mlaldim(\mathbf{X}) - \eta$.
	\end{enumerate}
\end{prop}

The proof is given at the end of Section~\ref{sec: laldim main}. Once it is proved, it is easy to deduce Theorem \ref{thm: laldim main} from Theorem \ref{thm: mdimcor detail}

\begin{proof}[{\bf Proof of Theorem \ref{thm: laldim main}}] Let $(\Om, \Fk, \PP)$ be an underlying probability space, on which both $\mathrm{X}$ and random matrices $A_n$ are defined (recall that we assume them to be independents). By translating the process by $\E X_1$, we can assume that $\E X_1 = 0$. Note further that is suffices to consider the case $\mlaldim(\mathbf{X}) > 0$ as otherwise the assumption $\liminf \limits_{n \to \infty} \frac{m_n}{n} < \mlaldim(\mathbf{X})$ cannot hold. Fix $\eta >0$ such that $\liminf \limits_{n \to \infty} \frac{m_n}{n} < \mlaldim(\mathbf{X}) - \eta$. Fix $M \geq 1$ such that $\Var(X_1)/M^2 < 1$ and consider the sequence $E_n$ from Proposition \ref{prop: restriction}. Applying Theorem \ref{thm: mdimcor detail} to the sequence $(\mu_n|_{E_n})$ we have that for all $\delta$ small enough
	\[ \liminf \limits_{n \to \infty}\ \mu_n \otimes G_{n,m_n} \left( \left\{ (x, A) \in E_n \times \R^{m_n \times n} : \frac{1}{\sqrt{n}}\|x - F_n(Ax, A)\|_2 \leq \delta  \right\} \right) = 0,\]
	so for such $\delta$
	\[ \liminf \limits_{n \to \infty} \PP\left( \frac{1}{\sqrt{n}} \|X^n - \hat{X}^n\|_2 \leq \delta \right) \leq \limsup \limits_{n \to \infty}\ (1 - \mu_n(E_n)) \leq  \Var(X_1)/M^2 < 1.\]
	Therefore $\frac{1}{\sqrt{n}} \|X^n - \hat{X}^n\|_2$ cannot converge in probability to zero.
\end{proof}

The rest of this section is devoted to the proof of Proposition \ref{prop: restriction}. For $k \in \N$ let us denote $\ld_k(x) = \ld(\mu_k, x)$ for short. The idea for the construction of $E_n$ is to consider the set of points $x$ at which the bound $\mu_k(B_2^k(x,r)) \lesssim r^{\ld_k(x)}$ holds uniformly in $r > 0$ with suitable constants (see Lemma~\ref{lem: ball Sn bound} below). To obtain this uniformity, we introduce a truncation via  auxiliary functions $C_{\eps,k}(x)$ and $f_{C,\eps,k}$ below. Given $x \in \R^k$ and $\eps>0$ set
\[ C_{\eps,k}(x) = \begin{cases} \sup \limits_{r > 0} \frac{\mu_k(B_2^k(x,r))}{r^{\ld_k(x) - \eps}} & \text{ if } \ld_k(x) \geq \eps \\
	\infty & \text{otherwise}
\end{cases} \]
Note that $C_{\eps,k}(x) < \infty$ whenever $\ld_k(x) \geq \eps$ and then $\mu_k(B_2^k(x,r)) \leq C_{\eps,k}(x) r^{\ld_k(x) - \eps}$ holds for all $r>0$. Given $C \geq 1$ and $0 < \eps < 1/2$ and $k \in \N$ let us define an auxiliary function $f_{C, \eps,k} : \R^k \to [0, \infty)$
\[ f_{C, \eps, k}(x) = (\ld_k(x) - \eps) \mathds{1}_{\{C_{\eps,k}(x) \leq C,\ 2\eps \leq \ld_k(x) \leq C\}}(x). \]
Note that by Lemma \ref{lem: psi star mlaldim}
\begin{equation}\label{eq: f int}
	\lim \limits_{k \to \infty}\ \lim \limits_{\eps \to 0}\ \lim \limits_{C \to \infty}\ \frac{1}{k}\int f_{C,\eps,k}d\mu_k = \lim \limits_{k \to \infty} \frac{1}{k} \int \ld_k(x)d\mu_k(x) = \mlaldim(\mathbf{X})
\end{equation}
(the inner equality holds by the monotone convergence theorem, since $\lim \limits_{\eps \to 0} \lim \limits_{C \to \infty} f_{C,\eps,k}(x) = \ld_k(x)$ monotonically for every $x$) and limits in $\eps$ and $C$ are increasing. Fix $g \in \N$ such that $\psi^*(g) < \infty$ and set $m:= k + g - 1$. For $n \geq 1$ define a function $S_{n,C,\eps, k} : \R^n \to [0,\infty)$
\[ S_{n,C,\eps, k} (x_1, \ldots, x_n) = \sum \limits_{j=0}^{\lfloor n / m \rfloor - 1} f_{C,\eps,k}(x_{jm + g}, x_{jm + g + 1 }, \ldots, x_{(j+1)m}). \]
We will use it later to define sets $E_n$ in Proposition \ref{prop: restriction}. One hand, it connects via the ergodic theorem and \eqref{eq: f int} to $\mlaldim(\mathbf{X})$. On the other hand, the following lemma shows that it controls measures of $n$-balls (uniformly in the radius), and hence it can be used the bound the energy integrals. In order to make of the $\psi^*$-mixing condition, we consider in $S_{n,C,\eps,k}$ blocks of coordinates which are $g$-separated (so heuristically we can treat the elements of the sum as essentially independent).
\begin{lem}\label{lem: ball Sn bound}
	Fix $g \in \N$ such that $\psi^*(g) < \infty$. Then for every $C\geq 1,\ 0 < \eps < 1/2,\ n, k \geq 1,\ r>0$ and $x \in \R^n$
	\begin{equation}\label{eq: product measure bound}
		\mu_n(B_2^n(x, r)) \lesssim^e_{C, \eps,k , g} S_{n, C, \eps,k}(x)^{-S_{n, C, \eps, k}(x)/2}r^{S_{n, C, \eps, k}(x)}
	\end{equation}
	(we use here the convention $0^0 = 1$).
\end{lem}

\begin{proof}
	The proof is (essentially) by induction on $n$. Fix $C \geq 1,\ \eps > 0, k \geq 1$ and denote for short $f = f_{C,\eps, k},\ S_n = S_{n,C,\eps, k}$. Note that for every $x \in \R^k$ we have
	\begin{equation}\label{eq: one dim f bound}
		\mu_k(B^k_2(x,r)) \leq Cr^{f(x)} \text{ for all } r>0.
	\end{equation}
	Indeed, if $C_{\eps,k}(x) \leq C$ and $2\eps \leq \ld_k(x) \leq C$, then \eqref{eq: one dim f bound} follows from the definition of $C_{\eps,k}(x)$ and $f$. Otherwise $f(x) = 0$ and hence \eqref{eq: one dim f bound} holds since $C \geq 1$ and $\mu_k$ is a probability measure. For $x = (x_1, \ldots, x_n) \in \R^n$ define
	\[ D_n(x) = \sup \limits_{r>0}\ \frac{\mu_n(B_2^n(x, r))}{r^{S_n(x)}}. \]
	With this notation, our goal is to prove
	\begin{equation}\label{eq: Dn lemma}
		D_n(x) \lesssim^e_{C, \eps, k, g} S_{n, C, \eps,k}(x)^{-S_{n, C, \eps, k}(x)/2}.
	\end{equation}
	
	Let $n = \ell m + q$ with $\ell \in \N$ and $0 \leq q < m$ (so that $\ell$ is the number of terms in the sum defining $S_{n,C,\eps, k}$) and note that $S_n(x) = S_{\ell m}(x) = S_{(\ell-1) m}(x) + f\left(x_{(\ell - 1)m + g}^{\ell m}\right)$. Note also that if $S_n(x) = 0$, then $r^{S_n(x)} = 1$ and as $\mu_n$ is a probability measure it follows
	\[ S_n(x) = 0\ \Longrightarrow\ D_n(x) \leq 1, \]
	which proves \eqref{eq: Dn lemma} if $S_n(x) = 0$. Therefore, it suffices to consider the case $S_n(x) > 0$. 
	Assume first that $S_{(\ell - 1) m }(x)=0$. Then applying \eqref{eq: one dim f bound} (together with the stationarity of the process and the fact that if $y \in B_2^n(x, r)$, then $y_{(\ell-1)m + g}^{\ell m} \in B_2^k(x_{(\ell-1)m + g}^{\ell m}, r)$) gives
	\[ \mu_n(B_2^n(x, r))  \leq \mu_k(B_2^k(x_{(\ell-1)m + g}^{\ell m}, r))  \leq Cr^{f \left(x_{(\ell-1)m + g}^{\ell m }\right)}.\]
	As $S_n(x) = S_{\ell m}(x) = f \left(x_{(\ell-1)m + g}^{\ell m}\right)$ in this case, we have
	\begin{equation}\label{eq: Sn positive Sn-1 zero} S_n(x) > 0 \text{ and } S_{(\ell - 1)m}(x) = 0\ \Longrightarrow\ D_n(x) \leq C.
	\end{equation}
	Furthermore we have then $S_n(x) = f\left(x_{(\ell-1)m + g}^{\ell m}\right)\in [\eps,  C]$, so
	\[ S_n(x)^{-S_n(x) / 2} = f \left(x_{(\ell-1)m + g}^{\ell m}\right)^{-f \left(x_{(\ell-1)m + g}^{\ell m}\right) / 2} \geq Q\] for some constant $Q = Q(C, \eps) > 0$ and so \eqref{eq: Dn lemma} holds if $S_{n}(x) > 0$ and $S_{(\ell-1)m}(x) = 0$.
	
	It remains to consider the case with $S_{(\ell-1)m}(x)>0$ (and therefore $S_n(x) > 0$). We shall give a bound on $D_{n}(x)$ in terms of $D_{(\ell - 1)m}(x)$. Iterating this bound will yield \eqref{eq: Dn lemma}. Applying Lemma \ref{lem: psi star on balls} (for $i = (\ell-1)m$) gives

	\[\begin{split}
		& \mu_{n}  (B_2^{n}(x, r)) \leq \mu_{\ell m}(B_2^{\ell m}(x_1^{\ell m}, r))\\
		& \overset{\text{Lem \ref{lem: psi star on balls}}}{\leq} \psi^*(g) \int \limits_{B_2^{k}(x_{(\ell - 1)m + g}^{\ell m}, r)} \mu_{(\ell-1)m}\left(B_2^{(\ell-1)m}\left(x_1^{(\ell-1) m}, \left( r^2 - |x_{(\ell - 1)m + g}^{\ell m} - z|^2\right)^{1/2}\right)\right)d\mu_{k}(z) \\
		& \overset{\text{def. of } D_{(\ell-1)m}(x)}{\leq} \psi^*(g) D_{(\ell - 1)m}(x)   \int \limits_{B_2^{k}(x_{(\ell - 1)m + g}^{\ell m}, r)} \left( r^2 - |x_{(\ell - 1)m + g}^{\ell m} - z|^2\right)^{S_{(\ell-1)m}(x)/2}d\mu_k(z) \\
		& \overset{\text{\cite[Th.~1.15]{mattila}}}{=} \psi^*(g) D_{(\ell - 1)m}(x) \int \limits_{0}^\infty \mu_k \left( \left\{ z \in B_2^{k}(x_{(\ell - 1)m + g}^{\ell m}, r) :  \left( r^2 - |x_{(\ell - 1)m + g}^{\ell m} - z|^2\right)^{S_{(\ell-1)m}(x)/2} \geq t \right\}\right) dt \\
		& =  \psi^*(g)  D_{(\ell - 1)m}(x)  \int \limits_{0}^{r^{S_{(\ell - 1)m}(x)}} \mu_{k} \left( \left\{ z \in B_2^{k}(x_{(\ell - 1)m + g}^{\ell m}, r) :  |x_{(\ell - 1)m + g}^{\ell m} - z| \leq \sqrt{r^2 - t^{2 / S_{(\ell-1)m}(x)}} \right\}\right) dt \\
		& = \psi^*(g)  D_{(\ell - 1)m}(x)  \int \limits_{0}^{r^{S_{(\ell - 1)m}(x)}} \mu_{k} \left(  B_2^{k}\left(x_{(\ell - 1)m + g}^{\ell m},\sqrt{r^2 - t^{2 / S_{(\ell-1)m}(x)}} \right)  \right) dt \\
		& \overset{s = \frac{1}{r}t^{1/S_{(\ell-1)m}(x)}}{=}  \psi^*(g) D_{(\ell - 1)m}(x)   S_{(\ell-1)m}(x) r^{S_{(\ell-1)m}(x)} \int\limits_0^1s^{S_{(\ell-1)m}(x) - 1} \mu_{k} \left(  B_2^{k}\left(x_{(\ell - 1)m + g}^{\ell m}, r\sqrt{1 - s^2} \right)  \right)  ds \\
		& \overset{\eqref{eq: one dim f bound}}{\leq} C  \psi^*(g) D_{(\ell -1)m}(x) S_{(\ell-1)m}(x) r^{S_{(\ell-1)m}(x) + f\left( x_{(\ell-1)m+g}^{\ell m} \right)}  \int\limits_0^1 s^{S_{(\ell-1)m}(x) - 1} (1-s^2)^{f\left( x_{(\ell-1)m+g}^{\ell m} \right)/2} ds \\
		& \overset{u=s^2}{=} \frac{C  \psi^*(g)}{2}  D_{(\ell -1)m}(x) S_{(\ell-1) m}(x) r^{S_{n}(x)} \int\limits_0^1 u^{\frac{S_{(\ell - 1)m}(x)}{2} - 1} (1-u)^{f\left( x_{(\ell-1)m+g}^{\ell m} \right)/2} du \\
		& = \frac{C  \psi^*(g)}{2} D_{(\ell -1)m}(x) S_{(\ell-1) m}(x) r^{S_{n}(x)} B\left(\frac{S_{(\ell-1)m}(x)}{2}, \frac{f\left( x_{(\ell-1)m+g}^{\ell m} \right)}{2} + 1\right) \\
		& \overset{\eqref{eq: gamma beta}}{=} \frac{C  \psi^*(g)}{2} D_{(\ell -1)m}(x)  S_{(\ell-1) m}(x)  r^{S_{n}(x)} \frac{  \Gamma\left( \frac{S_{(\ell-1)m}(x)}{2}\right) \Gamma\left( \frac{f\left( x_{(\ell-1)m+g}^{\ell m} \right)}{2}  + 1\right) }{\Gamma  \left( \frac{S_{\ell m}(x)}{2} + 1 \right)}.
	\end{split} \]
	
	Consequently
	
	\begin{equation}\label{eq: Dn iterate}
		D_{n}(x) \leq  \frac{C  \psi^*(g)}{2} D_{(\ell -1)m}(x)  S_{(\ell-1) m}(x)\frac{  \Gamma\left( \frac{S_{(\ell-1)m}(x)}{2}\right) \Gamma\left( \frac{f\left( x_{(\ell-1)m+g}^{\ell m} \right)}{2}  + 1\right) }{\Gamma  \left( \frac{S_{\ell m}(x)}{2} + 1 \right)} \quad \text{ if } S_{(\ell-1)m}(x)>0.
	\end{equation}

	Let $\ell_0 = \inf \{ 1 \leq j \leq \ell -1 : S_{j m}(x) > 0 \}$. Iterating \eqref{eq: Dn iterate} gives
	
	\[ D_n(x) \leq D_{\ell_0 m}(x) \left( \frac{C  \psi^*(g)}{2} \right)^{\ell - \ell_0} \prod \limits_{j = \ell_0}^{\ell - 1} \left(S_{jm}(x) \frac{ \Gamma\left( \frac{S_{jm}(x)}{2}\right) \Gamma\left( \frac{f\left( x_{jm+g}^{(j+1)m} \right)}{2} + 1 \right) }{\Gamma  \left( \frac{S_{(j+1) m}(x)}{2} + 1 \right)} \right) \]
	
	As $\frac{f\left( x_{jm+g}^{(j+1)m} \right)}{2} + 1 \in [1, 1 + C/2]$, we have that  $\Gamma\left( \frac{f\left( x_{jm+g}^{(j+1)m} \right)}{2} +  1 \right) \leq Q$ for some constant $Q$ depending on $C$. Applying this and rearranging the product gives

	\begin{equation}\label{eq: Dn prod bound}
		D_{n}(x)  \leq  D_{\ell_0 m}(x)  \left(\frac{C  \psi^*(g) Q}{2}\right)^{\ell - \ell_0} \frac{S_{\ell_0 m}(x) \Gamma\left( \frac{S_{\ell_0 m}(x)}{2}  \right) }{\Gamma\left( \frac{S_{\ell m}(x)}{2} + 1  \right)} \prod \limits_{j = \ell_0 + 1}^{\ell - 1} \left( S_{jm}(x) \frac{\Gamma\left( \frac{S_{jm}(x)}{2}  \right)}{\Gamma \left( \frac{S_{jm}(x)}{2} + 1 \right)} \right).
	\end{equation}
	
	By \eqref{eq: gamma frac}, $\Gamma\left( \frac{S_{jm}(x)}{2} + 1 \right) = \frac{S_{jm}(x)}{2} \Gamma\left( \frac{S_{jm}(x)}{2} \right)$, so each factor in the product in \eqref{eq: Dn prod bound} equals $S_{jm}(x) \cdot \frac{1}{S_{jm}(x)/2} = 2$. Hence the product equals $2^{\ell - \ell_0 - 1}$.

	Note that by the definition of $\ell_0$ and the fact $f(x) > 0 \Rightarrow \eps \leq f(x) \leq C$, we have
	\[
	\eps \leq S_{\ell_0 m}(x) \leq C \text{ and } \eps \leq S_{jm}(x) \leq C\ell \text{ for } \ell_0 \leq j \leq \ell. \]
	Combining this with \eqref{eq: gamma bounds} and \eqref{eq: gamma frac} gives that there exist constants $R_1 = R_1(C,\eps)\text{ and } R_2 = R_2(C, \eps)$ such that
	\begin{itemize}
		\item $S_{\ell_0 m}(x) \Gamma \left( \frac{S_{\ell_0 m}(x)}{2} \right) \overset{\eqref{eq: gamma bounds}}{\leq} R_1$,
		\item $\Gamma\left( \frac{S_{\ell m}(x)}{2} + 1  \right) \overset{\eqref{eq: gamma frac}}{=} \frac{S_{\ell m}(x)}{2} \Gamma\left( \frac{S_{\ell m}(x)}{2} \right) \overset{\eqref{eq: gamma bounds}}{\geq} L_\eps \left(\frac{S_{\ell m}(x)}{2}\right)^{S_{\ell m}(x)/2 + 1/2} e^{-S_{\ell m}(x)/2} \geq R_2^{\ell} S_{\ell m}(x)^{S_{\ell m}(x) / 2}\cdot e^{-S_{\ell m}(x)/2}$.
	\end{itemize}
	where the last inequality uses $\left(\frac{S_{\ell m}(x)}{2}\right)^{1/2} \geq (\eps/2)^{1/2}$ and $2^{-S_{\ell m}(x)/2} \geq (2^{-C/2})^\ell$. Applying the above inequalities to \eqref{eq: Dn prod bound} gives for some constant $R_3 = R_3(C,\eps,g) \geq 1$
	\[ D_{n}(x)  \leq  D_{\ell_0 m}(x) R_3^{\ell} S_{\ell m}(x)^{-S_{\ell m}(x)/2}\cdot e^{S_{\ell m}(x)/2} = D_{\ell_0 m}(x) R_3^{\ell} S_{n}(x)^{-S_{n}(x)/2}\cdot e^{S_n(x)/2}, \]
	since $S_{\ell m}(x) = S_n(x)$. Setting $P = R_3 e^{C/2}$ and using $e^{S_n(x)/2} \leq (e^{C/2})^\ell$ (since $S_n(x) \leq C\ell$), we obtain
	\[ D_n(x) \leq D_{\ell_0 m}(x) P^\ell S_n(x)^{-S_n(x)/2}. \]
	As $\ell \leq n$, in order to obtain \eqref{eq: Dn lemma} and finish the proof of the lemma, we shall prove that $D_{\ell_0 m}(x) \leq C$. If $\ell_0 \geq 2$, then this follows from \eqref{eq: Sn positive Sn-1 zero}, as we have $S_{\ell_0 m}(x) > 0$ and $S_{(\ell_0 - 1)m}(x) = 0$ by the definition of $\ell_0$. If $\ell_0= 1$, then $D_{\ell_0 m}(x) = D_m(x) \leq C$, since $B_2^m(x,r) \subseteq \R^{g-1} \times B_2^k((x_g,\ldots,x_m),r)$, so we have
	\[\mu_m(B_2^m(x,r)) \leq \mu_k(B_2^k((x_g,\ldots,x_m),r)) \overset{\eqref{eq: one dim f bound}}{\leq} C r^{f(x_g,\ldots,x_m)} = Cr^{S_m(x)},\]
	where the last equality holds because $S_m(x) = f(x_g,\ldots,x_m)$ (the sum defining $S_m$ has a single term).
\end{proof}

\begin{proof}[{\bf Proof of Proposition \ref{prop: restriction}}]
	Fix $M \geq 1$ and $0 < \eta < 1$. Fix $g \in \N$ such that $\psi^*(g)< \infty$. By \eqref{eq: f int}, we can fix $k\in \N, C \geq 1$ and $0 < \eps < 1/2$ such that (recall that $f_{C,\eps,k}(x) \leq k$ for $\mu_k$-a.e. $x \in \R^k$, which follows from $\ld_k(x) \leq k$ $\mu_k$-a.e.\ (by \eqref{eq: loc dim leq n}) and the definition of $f_{C,\eps,k}$)
	\[ \frac{1}{k+g-1}\int f_{C,\eps,k} d\mu_k \geq \frac{1}{k}\int f_{C,\eps,k} d\mu_k - \eta/8 \geq \mlaldim(\mathbf{X}) - \eta/4.\]
	Set $m = g+k-1$. Let $\mu$ denote the distribution of $\mathbf{X}$ on $\R^\N$. By Birkhoff's pointwise ergodic theorem \cite[Theorem~1.14]{W82}\footnote{The ergodic theorem is applied to the function $\tilde{f}(x) = f_{C,\eps,k}(x_g, \ldots, x_m)$ and the measure-preserving transformation $\sigma^m$, where $\sigma$ denotes the left-shift on $\R^\N$, so that $S_{n,C,\eps,k}(x) = \sum_{j=0}^{\lfloor n/m \rfloor - 1} \tilde{f}(\sigma^{jm}(x))$. The ergodic theorem is applied to the $m$-th iterate $\sigma^m$ of the left-shift map, for which $\mu$ is ergodic, since $\mathbf{X}$ is $\psi^*$-mixing and therefore $\mu$ is strong-mixing (in the sense of \cite{W82}; see \cite[Section~2.1]{BradleySurvey05}), and every iterate of a strong-mixing transformation is ergodic - this follows for instance from \cite[Theorem~1.23]{W82}).}, for $\mu$-a.e.\ $x \in \R^\N$,
	\[ \lim \limits_{n \to \infty} \frac{1}{n} S_{n,C,\eps,k}(x) = \frac{1}{m} \int \limits_{\R^m} f_{C,\eps, k}(x_g, \ldots, x_{m})d\mu_m(x_1, \ldots, x_m) = \frac{1}{k+g-1} \int \limits_{\R^k} f_{C,\eps,k}  d\mu_k \geq \mlaldim(\mathbf{X}) - \eta/4.  \]
	Therefore
	\begin{equation}\label{eq: tilde En meas}
		\lim \limits_{n \to \infty} \mu_n \left( \left\{ x \in \R^n : \frac{1}{n} S_{n,C, \eps, k}(x) \geq \mlaldim(\mathbf{X}) - \eta/2 \right\} \right) = 1.
	\end{equation}
	By the Chebyshev's inequality (recall that we assume $\E X_1 = 0$)
	\begin{equation}\label{eq: chebyshev}  \mu_n(\R^n \setminus B_2^n(0, \sqrt{n}M)) = \mu_n \left(  \left\{ (x_1, \ldots, x_n) \in \R^n : \frac{1}{n}\sum \limits_{j=1}^n x_j^2 \geq M^2 \right\}\right) \leq  \frac{\Var(X_1)}{M^2}.\end{equation}
	Therefore, setting
	\[ E_n =   \left\{ x \in \R^n : \frac{1}{n} S_{n,C, \eps, k}(x) \geq \mlaldim(\mathbf{X}) - \eta/2 \right\} \cap B_2^n(0, \sqrt{n}M)\]
	we see from \eqref{eq: tilde En meas} and \eqref{eq: chebyshev} that $E_n$ satisfies items  \eqref{it: En ball} and \eqref{it: En measure bound} of Proposition \ref{prop: restriction}. It suffices to prove that \eqref{it: En mdimcor} is satisfied as well. Denote $\nu_n = \mu_n |_{E_n},\ d = \mlaldim(\mathbf{X}) \leq 1$ (recall Lemma \ref{lem: aldim mid idimr}) and $S_n(x) = S_{n,C,\eps,k}$. Fix $0 < \theta < d - \eta$. Thus $\theta + \frac{\eta}{2} < d -\frac{\eta}{2}$. Then, by the definitions of $E_n$, $f_{C,\eps,k}$ (which satisfies $f_{C,\eps,k} \leq C$), and $S_n$ (which has at most $n$ terms),
	\begin{equation}\label{eq: Sn theta bound}  (\theta + \frac{\eta}{2})n \leq S_n(x) \leq Cn \text{ for } x \in E_n.
	\end{equation}
	By \eqref{eq: energy integral}
	\[
	\begin{split} \mE_{\theta n} (\nu_n) &= \theta n \int \int \limits_{0}^\infty r^{-\theta n - 1} \nu_n(B_2^n(x,r)) dr d\nu_n(x) \\
		& = \theta n \int \int \limits_{0}^{\sqrt{n}} r^{-\theta n - 1} \nu_n(B_2^n(x,r)) dr d\nu_n(x) + \theta n \int \int \limits_{\sqrt{n}}^\infty r^{-\theta n - 1} \nu_n(B_2^n(x,r)) dr d\nu_n(x) \\
		& \overset{\text{Lem. } \ref{lem: ball Sn bound}}{\lesssim^e_{C,\eps,k,g,\theta}} \int S_n(x)^{-S_n(x)/2} \int \limits_{0}^{\sqrt{n}} r^{S_n(x)-\theta n - 1} dr d\nu_n(x) + \int \limits_{\sqrt{n}}^\infty r^{-\theta n - 1}dr \\
		& = \int \frac{n^{(S_n(x) - \theta n)/2}}{S_n(x) - \theta n} S_n(x)^{-S_n(x)/2}  d\nu_n(x) + \frac{n^{-\theta n /2}}{\theta n} \\
		& \overset{\eqref{eq: Sn theta bound}}{\lesssim^e_{C,\eps,k,g,\theta} } \int n^{(S_n(x) - \theta n)/2} ((\theta + \eta/2)n)^{-S_n(x)/2}  d\nu_n(x) + n^{-\theta n /2} \\
		& = n^{-\theta n /2} \int (\theta + \eta/2)^{-S_n(x)/2}  d\nu_n(x) + n^{-\theta n /2} \\
		& \overset{\theta + \eta/2 < 1\text{ and } \eqref{eq: Sn theta bound} }{\leq} n^{-\theta n /2} (\theta + \eta/2)^{-Cn/2}  + n^{-\theta n /2} \\
		& \lesssim^e_{C,\eps,k,g,\theta,\eta} n^{- \theta n /2}.
	\end{split}
	\]
	As $\theta < d -\eta$ can be chosen arbitrarily, this shows $\mdimcor( (\nu_n)_{n=1}^\infty) \geq d - \eta$ by \eqref{eq: mdimcor lesssim e}, since $C,\eps,k,g$ are fixed given $\eta$.
\end{proof}

\section*{Acknowledgments}
The authors are grateful to Shirin Jalali and Tobias Koch for helpful discussions. We are also very grateful to the editor and anonymous referees for many useful comments and suggestions.

\appendix

\section{Proof of Lemma \ref{lem: aldim to idim}}\label{sec: lem aldim to idim proof}
Inequality $\ualdim(\mu) \leq n$ was obtained in \eqref{eq: aldim leq n}. Inequality $\laldim(\mu) \leq \lid (\mu)$ follows from Fatou's lemma. For $\uid(\mu) \leq \ualdim(\mu)$ we can also invoke Fatou's lemma for the upper limits: if $h_r \leq g$ $\mu$-a.e.\ for all $r$, where $g \in L^1(\mu)$, then $\limsup_{r \to 0} \int h_r \, d\mu \leq \int \limsup_{r \to 0} h_r \, d\mu$ (this follows by applying Fatou's lemma \cite[Theorem~1.28]{R87} to the nonnegative functions $g - h_r$). This requires checking that the collection of functions $x \mapsto \frac{\log \mu(B_2^n(x,r))}{\log r}$ is majorized by an integrable function. For that we shall use the assumption that $\mu$ has finite variance. Our goal is to prove that
\begin{equation}\label{eq: fatou int finite} \int \limits_{\supp(\mu)} \sup \limits_{0 < r \leq \frac{1}{5}} \frac{\log \mu(B_2^n(x,r))}{\log r} d\mu(x) < \infty.
\end{equation}
For $t \geq 0$ define
\[ A_t := \left\{ x \in \supp(\mu) : \underset{ 0 < r \leq \frac{1}{5}}{\exists} \frac{\log \mu(B_2^n(x,r))}{\log r} > t \right\} = \left\{ x \in \supp(\mu) : \underset{ 0 < r \leq \frac{1}{5}}{\exists}  \mu(B_2^n(x,r)) < r^t \right\}. \]
Then, by \cite[Theorem~1.15]{mattila},
\begin{equation}\label{eq: int to leb} \int \limits_{\supp(\mu)} \sup \limits_{0 < r \leq \frac{1}{5}} \frac{\log \mu(B_2^n(x,r))}{\log r} d\mu(x) = \int \limits_{0}^\infty \mu \left( \left\{ x \in \supp(\mu) : \sup \limits_{0 < r \leq \frac{1}{5}} \frac{\log \mu(B_2^n(x,r))}{\log r} > t \right\} \right)dt = \int \limits_{0}^\infty \mu \left( A_t \right)dt.\end{equation}
Consider a cover $A_t \cap B_2^n(0,t) \subset \bigcup \limits_{x \in A_t} B(x, r_x / 5)$, where $0 < r_x \leq \frac{1}{5}$ is such that $\mu(B(x, r_x)) < r_x^t$. By the Vitali $5r$-covering lemma (see e.g. \cite[Theorem 2.1]{mattila}) there exists at most countable set $E \subset A_t$ such that the family $\{ B(x, r_x / 5) : x \in E\}$ consist of pairwise disjoint sets and $A_t \cap B_2^n(0,t) \subset \bigcup \limits_{x \in E} B(x, r_x)$. We have for $t > n$
\begin{equation}\label{eq: cover rx}\begin{split}
		\mu(A_t \cap B_2^n(0,t)) & \leq \sum \limits_{x \in E} \mu(B(x, r_x)) \leq \sum \limits_{x \in E} r_x^t \leq 5^{n-t} \sum \limits_{x \in E} r_x^n \leq 5^{2n-t} \sum \limits_{x \in E} (r_x/5)^n \\
		& = \frac{5^{2n-t}}{\alpha(n)} \sum \limits_{x \in E} \Leb_n(B_2^n(x, r_x / 5)) \\
		& \leq \frac{5^{2n-t}}{\alpha(n)} \Leb_n(B_2^n(0,t)) = 5^{2n-t} t^n.
	\end{split}
\end{equation}
On the other hand, Chebyshev's inequality gives for $t > 0$
\[ \mu(\R^n \setminus B_2^n(0, t)) \leq \frac{\int \|x\|_2^2d\mu(x)}{t^2}. \]
Combining this with \eqref{eq: int to leb} and \eqref{eq: cover rx} gives, as $\mu$ is a probability measure
\[\begin{split}
	\int \limits_{\supp(\mu)} \sup \limits_{0 < r \leq \frac{1}{5}} \frac{\log \mu(B_2^n(x,r))}{\log r} d\mu(x) & \leq n + \int \limits_{n}^\infty \mu(\R^n \setminus B_2^n(0, t))dt + \int \limits_{n}^\infty \mu(A_t \cap B_2^n(0,t)) dt \\
	& \leq n +  \int \limits_{n}^\infty \left( \frac{\int \|x\|_2^2d\mu(x)}{t^2} + 5^{2n-t} t^n  \right)dt < \infty.
\end{split}\]
This proves \eqref{eq: fatou int finite}.
Finally, If the local dimension exists at $\mu$-a.e. $x$, then by the already proved inequalities 
$\uid (\mu) \leq \ualdim(\mu) = \int d(\mu, x)d\mu(x) = \laldim(\mu) \leq \lid (\mu)$, hence $\aldim(\mu) = \idim(\mu)$. This finishes the proof of Lemma \ref{lem: aldim to idim}.

\section{Proof of Lemma \ref{lem: psi star mlaldim}}\label{sec: lem psi star mlaldim proof}

Denote $a_n = \laldim(X_1^n)$. We shall prove that if $\mathbf{X}$ is $\psi^*$-mixing and $g \in \N$ is such that $\psi^*(g)<\infty$, then
\begin{equation}\label{eq: gapped supadd}
	a_{n+k+g-1} \geq a_n + a_k \text{ for all } n,k \geq 1.
\end{equation}
If \eqref{eq: gapped supadd} holds, then $\lim \limits_{n \to \infty} \frac{a_n}{n}$ exists by the gapped version of Fekete's lemma, see \cite[Lemma 2.1]{RaquepasKingman} and hence the proof will be finished. Since $\|x\|_\infty \leq \|x\|_2 \leq \sqrt{n}\|x\|_\infty$ for all $x \in \R^n$, we have $B_2^n(x,r) \subseteq B_\infty^n(x,r) \subseteq B_2^n(x, \sqrt{n}\, r)$ for all $x \in \R^n$ and $r > 0$. Therefore,
\[ \liminf_{r \to 0} \frac{\log \mu_n(B_2^n(x,\sqrt{n}\,r))}{\log r} \leq \liminf_{r \to 0} \frac{\log \mu_n(B_\infty^n(x,r))}{\log r} \leq \liminf_{r \to 0} \frac{\log \mu_n(B_2^n(x,r))}{\log r}, \]
and since $\lim_{r \to 0} \frac{\log \sqrt{n}}{\log r} = 0$, both sides equal $\ld(\mu_n, x)$. Thus, for $x \in \R^n$
\[ \ld(\mu_n, x) = \liminf \limits_{r \to 0} \frac{\log \mu_n(B_\infty^n(x,r))}{\log r}. \]
Therefore, by the definition of $\psi^*(g)$, we have for for $x = (x_1, \ldots, x_{n+k+g-1}) \in \R^{n+k+g-1}$
\[\begin{split} \ld(\mu_{n+k+g-1}, x) & = \liminf \limits_{r \to 0} \frac{\log \mu_{n+k+g-1}(B_\infty^{n+k+g-1}(x,r))}{\log r} \\
	& = \liminf \limits_{r \to 0} \frac{\log \mu_{n+g+k-1} \left(B_\infty^n(x_1^n,r) \times B_\infty^{g-1}(x_{n+1}^{n+g-1},r) \times B_\infty^{k}(x_{n+g}^{n+k+g-1},r) \right)}{\log r} \\
	& \geq \liminf \limits_{r \to 0} \frac{\log \mu_{n+g+k-1} \left(B_\infty^n(x_1^n,r) \times \R^{g-1} \times B_\infty^{k}(x_{n+g}^{n+k+g-1},r) \right)}{\log r} \\
	& \geq \liminf \limits_{r \to 0} \frac{\log \left( \psi^*(g) \mu_{n} \left(B_\infty^n(x_1^n,r) \right) \mu_k \left(B_\infty^{k}(x_{n+g}^{n+k+g-1},r) \right) \right)}{\log r} \\
	& \geq \liminf \limits_{r \to 0} \frac{\log  \psi^*(g) }{\log r} +  \liminf \limits_{r \to 0} \frac{\log \mu_{n} \left(B_\infty^n(x_1^n,r) \right) }{\log r} +  \liminf \limits_{r \to 0} \frac{\log \mu_k \left(B_\infty^{k}(x_{n+g}^{n+k+g-1},r) \right)}{\log r} \\
	& = \ld (\mu_n, x_1^n) + \ld (\mu_k, x_{n+g}^{n+k+g-1}).
\end{split}\]
Consequently, as $\mathbf{X}$ is stationary
\[\begin{split}
	a_{n+k+g-1} & = \int \limits_{\R^{n+k+g-1}} \ld(\mu_{n+k+g-1}, x) d\mu_{n+g+k-1}(x) \\
	&\geq \int \limits_{\R^{n+k+g-1}} \ld (\mu_n, x_1^n)d\mu_{n+g+k-1}(x) +  \int \limits_{\R^{n+k+g-1}}  \ld (\mu_k, x_{n+g}^{n+k+g-1})d\mu_{n+g+k-1}(x) \\
	& = \int \limits_{\R^n} \ld(\mu_n, x) d\mu_n(x) + \int \limits_{\R^k} \ld(\mu_k, x) d\mu_k(x) \\
	& = a_n + a_k.
\end{split}\]
This proves \eqref{eq: gapped supadd} and finishes the proof of Lemma  \ref{lem: psi star mlaldim}.

\section{$\psi^*$-mixing and local dimension regularity of Markov chains with continuous states}\label{sec: markov psi star proof}

\begin{prop}\label{prop: markov psi star}
Let $\nu$ be a probability measure on~$\mathbb{R}$ and let
$(X_n)_{n\geq 1}$ be a stationary Markov chain on~$\mathbb{R}$
whose transition kernel has a density $f(x,y)$ with respect
to~$\nu \otimes \nu$ satisfying
\begin{equation}\label{eq: markov f bounds}
  0 < c \;\leq\; f(x,y) \;\leq\; C \;<\;\infty
  \qquad\text{for } \nu\text{-a.e.\ } x\in\mathbb{R} \text{ and }\nu\text{-a.e.\ } y\in\mathbb{R}.
\end{equation}
Then $(X_n)_{n\geq 1}$ is $\psi^*$-mixing with
\begin{equation}\label{eq: markov psi bound}
  \psi^*(g)\;\leq\; 1 + R\,\alpha^{g},
  \qquad g\geq 1,
\end{equation}
for some constants $R > 0$ and $\alpha \in (0,1)$.
Moreover, if $\nu$ is local dimension regular with
$\aldim(\nu)=d$, then all finite-dimensional marginals
of $(X_n)$ are local dimension regular, and
\begin{equation}\label{eq: markov dim}
  \mualdim(\mathbf{X})
  = \mlaldim(\mathbf{X})
  = \midim(\mathbf{X})
  = \idimr(\mathbf{X})
  = \aldim(\nu)
  = \idim(\nu)
  = d.
\end{equation}
\end{prop}

\begin{proof}
By modifying $f$ on a set of $\nu \otimes \nu$-measure zero, we can assume that inequalities \eqref{eq: markov f bounds} hold for all $x,y \in \supp(\nu)$. The transition probabilities of $(X_n)$ are given by $P(x,B)=\int_B f(x,y)\,d\nu(y)$. By~\eqref{eq: markov f bounds},
\begin{equation}\label{eq: markov doeblin}
  P(x,B) \;\geq\; c\,\nu(B)
  \qquad\text{for all } x\in\supp(\nu) \text{ and all Borel }
  B\subseteq\mathbb{R}.
\end{equation}
We shall apply \cite[Theorem~16.0.2]{MeynTweedie2009}. It is easy to check that the chain is aperiodic. By \eqref{eq: markov doeblin}, the Doeblin condition holds, hence \cite[Theorem~16.0.2]{MeynTweedie2009} yields that
there exists a unique stationary probability measure~$\mu$ and,
by~\cite[Theorem~16.0.2.(ii)]{MeynTweedie2009},
\begin{equation}\label{eq: markov tv decay}
  \sup_{x \in \supp(\nu)}\|P^g(x,\cdot)-\mu\|_{TV}\;\leq\; R\,\alpha^g,
  \qquad g\geq 0,
\end{equation}
for some constants $R > 0$ and $\alpha \in (0,1)$.

Since $P(x,\cdot)\ll\nu$ for every $x \in \supp(\nu)$,
iterating the Chapman--Kolmogorov equation shows that
$P^g(x,\cdot)\ll\nu$ for every $g\geq 1$ and $x\in \supp(\nu)$; denote
the density by $f^{(g)}(x,\cdot)$.
For the stationary measure we have $\mu(\cdot)=\int P(x,\cdot)\,\mu(dx)$,
so by Fubini and \eqref{eq: markov f bounds}, we have $\mu\ll\nu$ with density
\begin{equation}\label{eq: markov pi bounds}
  \pi(y)\;:=\;\int f(x,y)\,\mu(dx),
  \qquad c\leq\pi(y)\leq C \quad\nu\text{-a.e.}
\end{equation}
The Chapman--Kolmogorov equation gives
\[
  f^{(g)}(x,y)-\pi(y)
  = \int \bigl[f^{(g-1)}(x,z)-\pi(z)\bigr]\,f(z,y)\,d\nu(z),
\]
hence, using $f(z,y)\leq C$,
\begin{equation}\label{eq: markov ptwise}
  |f^{(g)}(x,y)-\pi(y)|
  \;\leq\; C\!\int |f^{(g-1)}(x,z)-\pi(z)|\,d\nu(z)
  = 2C\,\|P^{g-1}(x,\cdot)-\mu\|_{TV}.
\end{equation}
Set $R_g := \sup_{x\in\mathbb{R}}\|f^{(g)}(x,\cdot)/\pi(\cdot)\|_{L^\infty(\nu)}$.
By \eqref{eq: markov tv decay}, \eqref{eq: markov ptwise} and $\pi\geq c$ from~\eqref{eq: markov pi bounds},
\begin{equation}\label{eq: markov Rg}
  R_g \;\leq\; 1 + \frac{2C}{c}\,
  \sup_{x\in\supp(\nu)}\|P^{g-1}(x,\cdot)-\mu\|_{TV} \leq 1 + \frac{2CR\alpha^{g-1}}{c}.
\end{equation}
Fix $g,n\geq 1$ and events $A\in\sigma(X_1,\ldots,X_n)$,
$B\in\sigma(X_{n+g},X_{n+g+1},\ldots)$ with $\PP(A),\PP(B)>0$.
By the Markov property, $\PP(A\cap B)=\E[\mathbf{1}_A\,\PP(B\mid X_n)]$,
so for \eqref{eq: markov psi bound} it suffices to show
\begin{equation}\label{eq: markov goal}
  \PP(B\mid X_n=x) \;\leq\; R_g\,\PP(B)
  \qquad\text{for $\mu$-a.e.\ }x.
\end{equation}

For a cylinder set of the form
$B_m=\{X_{n+g}\in A_1,\,X_{n+g+1}\in A_2,\,\ldots,\,X_{n+g+m-1}\in A_m\}$
(with $A_1,\ldots,A_m$ Borel), the Markov property gives
\begin{align}
  \PP(B_m\mid X_n\!=\!x)
  &= \int_{A_1}\!f^{(g)}(x,y_1)
    \underbrace{\int_{A_2}\!\cdots\!\int_{A_m}
    f(y_1,y_2)\cdots f(y_{m-1},y_m)\,d\nu(y_m)\cdots d\nu(y_2)
    }_{=:\,h(y_1)\,\geq\,0}\;d\nu(y_1) \notag\\
  &\leq R_g\int_{A_1}\!\pi(y_1)\,h(y_1)\,d\nu(y_1)
  \;=\; R_g\,\PP(B_m). \label{eq: markov cylinder}
\end{align}

The collection $\mathcal{A}$ of all finite unions of cylinder sets is
an algebra that generates $\sigma(X_{n+g},X_{n+g+1},\ldots)$.
The bound~\eqref{eq: markov cylinder} extends to all $B\in\mathcal{A}$
(by noting that in fact every finite union of cylinders can be written as a finite union of disjoint cylinders). It extends further to general $B\in\sigma(X_{n+g},X_{n+g+1},\ldots)$ by approximating $B$ by $B' \in \mathcal{A}$ (for every $\varepsilon >0$ there exists $B' \in \mathcal{A}$ such that $\PP(B\setminus B' \cup B' \setminus B)<\varepsilon$, see
e.g.\ \cite[Lemma A.2.1]{Durrett2019}). This proves \eqref{eq: markov psi bound}.

Now we prove \eqref{eq: markov dim}. Suppose that $\nu$ is local dimension regular with $\aldim(\nu)=d$. Let $\mu_n$ denote the distribution of $(X_1,\ldots,X_n)$ on~$\mathbb{R}^n$. By the Markov property and stationarity,
\begin{equation}\label{eq: markov rn}
  \frac{d\mu_n}{d\nu^n}(x_1,\ldots,x_n)
  = \pi(x_1)\,f(x_1,x_2)\cdots f(x_{n-1},x_n)
  \;>\; 0
  \qquad\nu^n\text{-a.e.}
\end{equation}
since each factor is positive by~\eqref{eq: markov f bounds} and~\eqref{eq: markov pi bounds}. In particular $\mu_n\sim\nu^n$ ($\mu_n$ and $\nu^n$ are mutually absolutely continuous). Since $\nu$ is local dimension regular with $d(\nu,x)=d$ for
$\nu$-a.e.~$x$, the product measure $\nu^n$ is local dimension
regular with $d(\nu^n,\mathbf{x})=nd$ for $\nu^n$-a.e.\
$\mathbf{x}=(x_1,\ldots,x_n)\in\mathbb{R}^n$.
Indeed, for any $\mathbf{x}$ at which all coordinates have local
dimension~$d$, the ball inclusions
$\prod_{i=1}^n B(x_i,r/\sqrt{n})\subset B_2^n(\mathbf{x},r)
\subset\prod_{i=1}^n B(x_i,r)$
give
\[
  \prod_{i=1}^n\nu\bigl(B(x_i,r/\sqrt{n})\bigr)
  \;\leq\; \nu^n\bigl(B_2^n(\mathbf{x},r)\bigr)
  \;\leq\; \prod_{i=1}^n\nu\bigl(B(x_i,r)\bigr),
\]
and since $\frac{\log\nu(B(x_i,r/\sqrt{n}))}{\log r}
= \frac{\log\nu(B(x_i,r/\sqrt{n}))}{\log(r/\sqrt{n})}
\cdot\frac{\log(r/\sqrt{n})}{\log r}\to d\cdot 1$ as $r\to 0$,
we get $d(\nu^n,\mathbf{x})=nd$. Since $\mu_n\sim\nu^n$ with density~\eqref{eq: markov rn} that is
positive and finite $\nu^n$-a.e., the local dimension of $\mu_n$
at $\mathbf{x}$ coincides with that of $\nu^n$ at $\mu_n$-a.e.\
$\mathbf{x}$. Indeed, writing $\rho_n:=d\mu_n/d\nu^n>0$ for the Radon--Nikodym derivative,
\begin{equation}\label{eq: markov mu nu dim}
  \frac{\log\mu_n(B_2^n(\mathbf{x},r))}{\log r}
  = \frac{\log\nu^n(B_2^n(\mathbf{x},r))}{\log r}
    + \frac{\log\bigl(\mu_n(B_2^n(\mathbf{x},r))/
      \nu^n(B_2^n(\mathbf{x},r))\bigr)}{\log r}.
\end{equation}
At every Lebesgue point of $\rho_n$ (with respect to the
measure~$\nu^n$ and Euclidean balls),

\[ \mu_n(B_2^n(\mathbf{x},r))/\nu^n(B_2^n(\mathbf{x},r))
\to\rho_n(\mathbf{x})\in(0,\infty),\]
so the second term in \eqref{eq: markov mu nu dim} vanishes as $r\to 0$.
Therefore $d(\mu_n,\mathbf{x})=d(\nu^n,\mathbf{x})=nd$
at $\mu_n$-a.e.~$\mathbf{x}$,
and hence $\mu_n$ is local dimension regular with
$\aldim(\mu_n)=nd$. By Lemma~\ref{lem: aldim to idim}, this gives
$\aldim(\mu_n)=\idim(\mu_n)=nd$.
Since $(X_n)$ is $\psi^*$-mixing and all
finite-dimensional marginals are local dimension regular,
Lemma~\ref{lem: aldim idim} and Lemma~\ref{lem: psi star mlaldim}
give \eqref{eq: markov dim}.
\end{proof}

\section{Monotonicity properties of energy integrals}\label{sec: mdim cor}

\begin{lem}\label{lem: energy comparsion}
	For each $n \geq 1$, let $\mu_n$ be a subprobability measure on $\R^n$ and $0 \leq s_n  \leq \theta n$ a constant.  Let $\theta \geq 0$ be such that $\mE_{\theta n}(\mu_n) \lesssim^e_{\theta} n^{-\theta n/2} $. Then it holds
	\[ \mE_{s_n}(\mu_n) \lesssim^e_{\theta} n^{-s_n/2}. \]
\end{lem}
\begin{proof}
	\[
	\begin{split} \mE_{s_n}(\mu_n) & \leq \iint \limits_{\|x - y\|_2 \leq \sqrt{n}}\|x-y\|_2^{-s_n} d\mu_n(x)d\mu_n(y) + n^{-s_n/2} \iint \limits_{\|x - y\|_2 > \sqrt{n}} d\mu_n(x)d\mu_n(y) \\
		& \leq n^{(\theta n - s_n)/2} \iint \limits_{\|x - y\|_2 \leq \sqrt{n}}\|x-y\|_2^{-\theta n} d\mu_n(x)d\mu_n(y) + n^{-s_n/2} \mu_n(\R^n)^2 \\
		& \leq n^{(\theta n - s_n)/2} \mE_{\theta n}(\mu_n) + n^{-s_n/2}  \\
		& = n^{-s_n / 2} \left( n^{\theta n /2} \mE_{\theta n}(\mu_n) + 1 \right) \\
		& \lesssim^e_{\theta} n^{-s_n/2},
	\end{split} \]
	where the last inequality follows from $\mE_{\theta n}(\mu_n) \lesssim^e_{\theta} n^{-\theta n/2} $.
\end{proof}


\section{Random Gaussian matrices}\label{sec: gauss matrices}

Recall that $G_{n,m}$ denotes the standard Gaussian measure on $\R^{m \times n}$ with $m \leq n$ (see Section~\ref{sec: mdimcor main}).

\begin{lem}\label{lem: gauss trans}
	For every $n \geq m \geq 1$, $u \in \R^n \setminus \{ 0\}$, and $0< \eps < 1$
	\begin{equation}\label{eq: gauss trans} G_{n,m}(\{ A \in \R^{m \times n}: \|Au\|_2 \leq\eps \sqrt{m}\|u\|_2 \})  \leq e^m \eps^m,\end{equation}
		where $G_{n,m}$ is the standard Gaussian measure on $\R^{m_n \times n}$.
\end{lem}
\begin{proof}
	Let $A = [a_{i j}]_{(i,j) \in \{1, \ldots, 
	m \} \times \{ 1, \ldots, n\}}$, so that $a_{i j}$ are i.i.d. random variables with distribution $N(0,1)$ on a probability space $(\Om, \Fk, \PP)$. Denote $u = (u_1, \ldots, u_n)$ and observe
	\[ G_{n,m}(\{ A : \|Au\|_2 \leq \eps \sqrt{m}\|u\|_2 \}) = \PP \left( \sum \limits_{i=1}^m \left(\sum \limits_{j=1}^n a_{i j}u_j \right)^2 \leq \eps^2 m \|u\|_2^2  \right) = \PP \left( \sum \limits_{i=1}^m \left(\sum \limits_{j=1}^n \frac{a_{i j}u_j}{\|u\|_2} \right)^2 \leq \eps^2 m \right).\]
	Note that $Z_i = \sum \limits_{j=1}^n \frac{a_{i j}u_j}{\|u\|_2}$ are independent random variables with distribution $N(0,1)$. Therefore applying \cite[Lemma 2]{JalaliMalekiBaraniuk14} with $\tau = 1- \eps^2$ gives
	\[\begin{split} G_{n,m}(\{ A : \|Au\|_2 \leq \eps \sqrt{m}\|u\|_2 \}) & = \PP \left( \sum \limits_{i=1}^m Z_i^2 \leq m(1-\tau) \right) \leq e^{\frac{m}{2}(\tau + \log(1-\tau))} \\
		& =  e^{\frac{m}{2}(1 - \eps^2 + \log \eps^2)} = e^\frac{m(1 - \eps^2)}{2}\eps^m \\
		& \leq e^m \eps^m.
	\end{split}\]
\end{proof}

Second, we need a high probability bound on the operator norm $\|A\|$ (with respect to Euclidean norms) of a random Gaussian matrix $A \in \R^{m \times n}$ with $m \leq n$.

\begin{lem}\label{lem: gauss norm}
	There exists an absolute constant $K \geq 1$ such that for every $n \geq m \geq 1$
	\[ G_{n,m}( \{ A \in \R^{m \times n}:  \|A\| \geq K\sqrt{n} \} ) \leq 2e^{-n},\]
			where $G_{n,m}$ is the standard Gaussian measure on $\R^{m_n \times n}$.
\end{lem}

\begin{proof}
	Again, let $A = [a_{i j}]$ with $a_{i j}$ being i.i.d. $N(0,1)$ random variables over a probability space $(\Om, \Fk, \PP)$. By \cite[Theorem 4.4.5]{VershyninBook} (recall that we assume $m \leq n$), there exists a universal constant $C>0$ such that for all $t > 0$
	\begin{equation}\label{eq: norm conc} \PP( \{ A :  \|A\| \geq C(2\sqrt{n} + t)\max \limits_{i,j} \|a_{i j}\|_{\psi_2} \} ) \leq 2e^{-t^2},
	\end{equation}
	where $\| a_{i j} \|_{\psi_2}$ denotes the sub-Gaussian norm of the random variable $a_{i j}$ (see \cite[Definition 2.5.6]{VershyninBook}). By \cite[Example 2.5.8.(i)]{VershyninBook}, $\max \limits_{i,j} \|a_{i j}\|_{\psi_2}$ is bounded by an absolute constant (independently of $m,n$). Applying this together with \eqref{eq: norm conc} for $t = \sqrt{n}$ finishes the proof.
\end{proof}

\section{Conditional disintegration}\label{sec: cond}

We will need to work with conditional disintegration of measures with respect to linear maps. A useful formalism of this classical theory follows \cite{SimmonsRohlin}. For a Borel map $\phi \colon X \to \R^m$ on a compact set $X \subset \R^n$ and a (complete) finite Borel measure $\mu$ on $X$, we define a system of measures $\mu_{\phi, z}$, $z \in \R^m$, where $\mu_{\phi, z}$ is a (possibly zero) Borel measure on $\phi^{-1}(z)$ defined as the weak-$^*$ limit
\begin{equation}\label{eq:cond measure limit}
	\mu_{\phi,z} = \lim_{r \to 0} \frac{1}{\mu(\phi^{-1}(B_2^m(z,r)))} \; \mu|_{\phi^{-1}(B_2^m(z,r))},
\end{equation}
whenever the limit exists, and zero otherwise. By the topological Rohlin disintegration theorem \cite{SimmonsRohlin}, the limit in \eqref{eq:cond measure limit} exists for $\phi \mu$-almost every $z \in \R^m$ and satisfies
\begin{equation}\label{eq:cond meas decomp}
	\mu(E) = \int_{\R^m} \mu_{\phi, z}(E) \; d(\phi \mu)(z) \qquad\text{for  every }  \mu\text{-measurable } E \subset X
\end{equation}
(in particular, the function $\R^m \ni z \mapsto \mu_{\phi, z}(E)$ in \eqref{eq:cond meas decomp} is $\phi \mu$-measurable) and
\begin{equation}\label{eq:cond meas inverse}
	\mu_{\phi, z}(\phi^{-1}(z)) = 1 \qquad\text{for $\phi \mu$-almost every } z \in \R^m.
\end{equation}
The system $\{\mu_{\phi, z}\}_{z \in \R^m}$ is called the \textbf{system of conditional measures for $\mu$ with respect to $\phi$}. Moreover, the conditions \eqref{eq:cond meas decomp} and \eqref{eq:cond meas inverse} characterize the system  $\{\mu_{\phi, z}\}_{z \in \R^m}$ uniquely ($\phi \mu$-almost surely). See \cite{SimmonsRohlin} for details (note that \cite{SimmonsRohlin} considers only the case where $\mu$ is a probability measure, while in our case we consider a general finite measure $\mu$ and set the conditional measures $\mu_{\phi,z}$ to have (almost surely) unit mass. This case follows directly from \cite{SimmonsRohlin} by normalizing $\mu$ to be a probability measure).

We will also make use of the following simple observation. If $g \colon X \to [0, \infty]$ is lower semi-continuous, then for $\phi \mu$-almost every $z \in \R^k$,
\begin{equation}\label{eq:muG liminf}
	\int g\, d \mu_{\phi,z} \leq  \liminf_{r \to 0}  \frac{1}{\mu(\phi^{-1}(B_2^m(z,r)))} \int_{\phi^{-1}(B_2^m(z,r))} g\, d\mu.
\end{equation}
This follows from the definition of $\mu_{\phi,z}$ as a weak-$^*$ limit and the fact that a lower semi-continuous function $g \colon X \to [0, \infty]$ is a non-decreasing limit $g_k \nearrow g$ of a sequence of non-negative continuous functions $g_k : X \to [0, \infty)$ (or see e.g. \cite[Corollary 8.2.5]{BogachevMeasureTheory}). More precisely, by the monotone convergence theorem for non-negative functions (see e.g. \cite[Theorem 1.26]{R87})
\[
\begin{split} \int g\, d \mu_{\phi,z} & = \lim \limits_{k \to \infty} \int g_k d \mu_{\phi,z} = \lim \limits_{k \to \infty}\ \lim \limits_{r \to 0 } \frac{1}{\mu(\phi^{-1}(B_2^m(z,r)))} \int_{\phi^{-1}(B_2^m(z,r))} g_k\, d\mu \\
	& \leq \liminf \limits_{r \to 0 } \frac{1}{\mu(\phi^{-1}(B_2^m(z,r)))} \int_{\phi^{-1}(B_2^m(z,r))} g\, d\mu.
\end{split} \]

\section{Gamma and beta functions}\label{sec: gamma beta}

For $z > 0$ the gamma function is defined as
\[ \Gamma(z) = \int \limits_0^\infty t^{z-1}e^{-t} dt. \]
Recall that the gamma function extends the factorial function in the sense that $\Gamma(n) = (n-1)!$ for $n \in \N$. One can express the volume of the unit $n$-ball in its terms as
\begin{equation}\label{eq: vol unit ball}
	\alpha(n) := \Leb_n(B_2^n(0,1)) = \frac{\pi^{n/2}}{\Gamma(n/2 + 1)}.
\end{equation}
For $z_1, z_2 > 0$ the beta function is defined as
\[ B(z_1, z_2) = \int \limits_0^1 t^{z_1 - 1} (1-t)^{z_2 - 1} dt.\]

The two are connected via the following formula

\begin{equation}\label{eq: gamma beta}
	B(z_1, z_2) = \frac{\Gamma(z_1)\Gamma(z_2)}{\Gamma(z_1 + z_2)}.
\end{equation}

We will also make use of bounds, which follow directly from Stirling's approximation for the gamma function (see e.g. \cite[Eq. (3.9)]{ArtinGamma}):
\[ \Gamma(z) = \sqrt{\frac{2\pi}{z}}\left(\frac{z}{e}\right)^z\left(1 + O\left( \frac{1}{z} \right) \right), \]
as well of the fact that $\Gamma(z)$ is positive and continuous in $(0,\infty)$. It follows that there exists an absolute constant $L$ and constant $L_\eps$ depending on $\eps>0$ such that
\begin{equation}\label{eq: gamma bounds}  L_\eps z^{z - 1/2} e^{-z} \leq \Gamma(z) \leq L z^{z - 1/2} e^{-z} \text{ for } z \geq \eps.
\end{equation}

We will also use the identity
\begin{equation}\label{eq: gamma frac} \Gamma(z+1) = z\Gamma(z) \text{ for } z>0.
\end{equation}

\section{Covering bounds}\label{sec: covering lemmas}

The following is a variant of the Hardy-Littlewood maximal inequality of weak type (see~\cite[Theorem~2.19]{mattila}) with an explicit control of constants.

\begin{lem}\label{lem: density bound}
	Let $n \geq m \geq 1$ and $M \geq 1$. Let $\mu$ be a finite Borel measure on $\R^n$ such that $\mu(\R^n \setminus B_2^n(0,\sqrt{n}M)) = 0$. Then for every linear map $A \in \R^{m \times n}$ and every $D>0$
	\[ A\mu \left( \left\{ x \in \R^m : \underset{0 < r \leq 1}{\exists}\ \mu(A^{-1}(B_2^m(x,r))) \leq Dr^m  \right\} \right) \leq D (5\|A\|\sqrt{n}M + 1)^m. \]
\end{lem}

\begin{proof}
	First, note that
	\begin{equation}\label{eq: supp Amu} \supp A\mu \subset A(B_2^n(0, \sqrt{n}M)) \subset B_2^m(0, \|A\|\sqrt{n}M).
	\end{equation}
	Denote
	\[ E = \left\{ x \in \supp A\mu :  \underset{0 < r \leq 1}{\exists}\ \mu(A^{-1}(B_2^m(x,r))) \leq Dr^m \right\}\]
	and consider a cover
	\[ E \subset \bigcup \limits_{x \in E} B_2^m(x, r_x / 5), \]
	where $0 < r_x \leq 1$ is such that $A\mu(B_2^m(x,r_x)) \leq Dr_x^m$. By the Vitali $5r$-covering lemma (see e.g. \cite[Theorem 2.1]{mattila}) there exists at most countable set $F \subset E$ such that the family $\{B_2^m(x, r_x / 5) : x \in F \}$ consists of pairwise disjoint sets and $E \subset \bigcup \limits_{x \in F} B_2^m(x, r_x)$. Therefore
	\begin{equation}\label{eq: AmuE}
		A\mu(E) \leq \sum \limits_{x \in F} A\mu(B_2^m(x,r_x)) \leq \sum \limits_{x \in F} D r_x^m.
	\end{equation}

	On the other hand, by the disjointness of $\{B_2^m(x, r_x / 5) : x \in F \}$ we have
	\begin{equation}\label{eq: sum rx}
		\sum \limits_{x \in F} r_x^m = \frac{5^m}{\alpha(m)} \sum \limits_{x \in F} \Leb_m(B_2^m(x, \frac{r_x}{5})) = \frac{5^m}{\alpha(m)}  \Leb_m\left(\bigcup \limits_{x \in F} B_2^m(x, \frac{r_x}{5})\right).
	\end{equation}
	As $F \subset \supp A \mu$, we have by \eqref{eq: supp Amu}
	\[ \bigcup \limits_{x \in F} B_2^m(x, \frac{r_x}{5}) \subset B_2^m(0, \|A\|\sqrt{n}M + \frac{1}{5}),\]
	so \eqref{eq: sum rx} gives
	\[ \sum \limits_{x \in F} r_x^m  \leq (5\|A\|\sqrt{n}M + 1)^m. \]
	Combining this with \eqref{eq: AmuE} finishes the proof.
\end{proof}

\begin{lem}\label{lem: ball to energy}
	Let $\mu$ be a finite Borel measure on $\R^n$. Then for every $s>0, z \in \R^n$ and $r>0$
	\[ \mu(B_2^n(z,r)) \leq 2^{s/2} r^{s/2} \mE_s(\mu)^{1/2}. \]
\end{lem}
\begin{proof}
	\[ \mE_s(\mu) \geq \int \limits_{B_2^n(z,r)} \int \limits_{B_2^n(z,r)} \|x - y\|_2^{-s} d\mu(x) d\mu(y) \geq (2r)^{-s} \mu(B_2^n(z,r))^2. \]
\end{proof}

\section{Examples}\label{sec: maldim noneq}

Below we present examples showing that the assumption of local dimension regularity cannot be omitted in Lemma \ref{lem: aldim idim} and Example \ref{ex: i.i.d.}.

\begin{example}\label{ex: maldim noneq}
	Let $S \subset \N$ be such that
	\[\liminf \limits_{n \to \infty} \frac{\#(S \cap [1,n])}{n} = \liminf \limits_{n \to \infty} \frac{\#((\N \setminus S) \cap [1,n])}{n} = 0\] and 
	\[\limsup \limits_{n \to \infty} \frac{\#(S \cap [1,n])}{n} = \limsup \limits_{n \to \infty} \frac{\#((\N \setminus S) \cap [1,n])}{n} = 1\]
	(for instance one can take $S = \bigcup \limits_{n=0}^\infty[s_{2n}, s_{2n+1})$, where $(s_n)_{n \geq 0}$ is a strictly increasing sequence of natural numbers such that $\lim \limits_{n \to \infty} \frac{s_n}{s_{n+1} - s_n} = 0$). Let $\eps_j, j \geq 1$ be a sequence of i.i.d random variables such that $\PP(\eps_j = 0) = \PP(\eps_j = 1) = 1/2$. Define random variables $Y, Z$ as $Y = \sum \limits_{j \in S} \eps_j 2^{-j}$ and $Z = \sum \limits_{j \in \N \setminus S} \eps_j 2^{-j}$. Let $Y_i, Z_i, i \geq 1$ be a collection of independent random variables, such that $Y_i$ have the same distribution as $Y$ and $Z_i$ have the same distributions as $Z$. Set $\mathbf{Y} = (Y_1, Y_2, \ldots)$ and $\mathbf{Z} = (Z_1, Z_2, \ldots)$.
	
	\begin{enumerate}
		\item\label{it: conv comb process} Let $\Delta$ be a random variable independent of all $Y_i, Z_i$ and such that $\PP(\Delta = 0) = \PP(\Delta = 1) = 1/2$. Set $X_i = \Delta Y_i + (1-\Delta)Z_i$ and $\mathbf{X} = (X_1, X_2, \ldots)$ (so $\mathbf{X} = \Delta\mathbf{Y} + (1-\Delta)\mathbf{Z}$). Then
		\[ \mlaldim(\mathbf{X}) = 0 < 1/2 = \midim(\mathbf{X}) = \idimr(\mathbf{X}) < 1 = \mualdim(\mathbf{X}). \]
		This example shows that all equalities in Lemma \ref{lem: aldim idim} can be violated if the finite-dimensional distributions of the process are not local dimension regular. Note that $\mathbf{X}$ is a stationary, but non-ergodic process.
		\item\label{it: iid process} Let $\Delta_i$ be a sequence of random variables independent of all $Y_i, Z_i$ and such that $\PP(\Delta_i = 0) = \PP(\Delta_i = 1) = 1/2$. Set $X_i = \Delta_i Y_i + (1-\Delta_i)Z_i$ and $\mathbf{X} = (X_1, X_2, \ldots)$ (so $\mathbf{X}$ is an i.i.d. process). Then
		\[ \laldim(X_1) = 0, \ualdim(X_1) = 1\]
		and
		\[\mlaldim(\mathbf{X}) = \mualdim(\mathbf{X}) = \midim(\mathbf{X}) = \idimr(\mathbf{X}) = \idim(X_1) = 1/2. \]
		This example shows that for an i.i.d. process without local dimension regular $1$-dimensional distribution, the mean average local dimension does not have to coincide with the average local dimension of its $1$-dimensional margin (cf. Example \ref{ex: i.i.d.})
	\end{enumerate}
\end{example}

Let us now prove the formulas in Example \ref{ex: maldim noneq}. Given $b,n \in \N$, let $\Ck^n_b$ be the partition of $\R^n$ into cubes of side length $2^{-b}$, such that each $C \in \Ck^n_b$ is an $n$-fold product of intervals of the form $[\frac{\ell}{2^b}, \frac{\ell+1}{2^b})$ with $\ell \in \Z$. Given $x \in \R^n$, let $C^n_b(x)$ be the unique element of $\Ck^n_b$ containing $x$. We will make use of the following fact (see e.g. \cite[Proposition 3.20]{HochmanNotes}): for a finite Borel measure $\mu$ on $\R^n$
\begin{equation}\label{eq: local dim dyadic}
	\ld(\mu, x) = - \liminf \limits_{b \to \infty} \frac{\log \mu(C^n_b(x))}{b} \text{ and } \ud(\mu, x) = - \limsup \limits_{b \to \infty} \frac{\log \mu(C^n_b(x))}{b} \text{ for } \mu\text{-a.e. } x.
\end{equation}
Note also that given a random vector $X^n$ taking values in $\R^n$ with distribution $\mu$, we have
\begin{equation}\label{eq: entropy vec part}
	H([X^n]_{2^b}) = H(\mu, \Ck^n_b),
\end{equation}
where $H(\mu, \Pk) = -\sum \limits_{C \in \Pk} \mu(C) \log \mu(C)$ is the entropy of $\mu$ with respect to a partition $\Pk$ of $\R^n$. We shall also use the fact that given two finite measure $\mu, \nu$ on $\R^n$ one has (see \cite[Corollary 3.17]{HochmanNotes})
\begin{equation}\label{eq: ac dim}
	\ld(\nu, x) = \ld(\mu, x) \text{ and } \ud(\nu, x) = \ud(\mu, x) \text{ at } \nu\text{-a.e. x, if } \nu \ll \mu.
\end{equation}

Consider now random variables $Y,Z,\Delta$ as defined in Example \ref{ex: maldim noneq}, with an underlying probability space $(\Om, \Fk, \PP)$. Let $\mu_Y, \mu_Z$ denote the distributions of $Y,Z$ on $\R$, respectively. Note that for each $b \in \N$, and $C \in \Ck^1_b$ of the form $C =[\frac{\ell}{2^b}, \frac{\ell+1}{2^b}), \ell \in \Z $ one has
\[\mu_Y(C) = 0\ \text{ or }\ \mu_Y(C) = \PP\left( \frac{\ell}{2^b} = \sum \limits_{j=1}^b \eps_j2^{-j} \mathds{1}_S(j)\right) = 2^{-\#(S \cap [1,b])}
\]
(following from the uniqueness of the dyadic expansion of integers)
and similarly
\[\mu_Z(C) = 0\ \text{ or }\ \mu_Z(C) = \PP\left( \frac{\ell}{2^b} = \sum \limits_{j=1}^b \eps_j2^{-j} \mathds{1}_{\N \setminus S}(j)\right) = 2^{-\#((\N \setminus S) \cap [1,b])}.\]
Note that it follows that all $C \in C^1_b$ with non-zero $\mu_Y$-measure are of equal measure (and likewise for $\mu_Z$). Therefore for all $b \geq 1$
\begin{equation}\label{eq: loc dim Y} \mu_Y(C^1_b(x)) = 2^{-\#(S \cap [1,b])} \text{ for } \mu_Y\text{-a.e. }x
\end{equation}
and 
\begin{equation}\label{eq: loc dim Z} \mu_Z(C^1_b(x)) = 2^{-\#((\N \setminus S) \cap [1,b])} \text{ for } \mu_Z\text{-a.e. }x.
\end{equation}

By \eqref{eq: loc dim Y} and \eqref{eq: loc dim Z}, combined with \eqref{eq: local dim dyadic} and assumption on $S$, we see that for $\mu_Y$-a.e. $x$
\[ \ld(\mu_Y, x) = \liminf \limits_{n \to \infty} \frac{\#(S \cap [1,b])}{b} = 0,\  \ud(\mu_Y, x) = \limsup \limits_{n \to \infty} \frac{\#(S \cap [1,b])}{b} = 1 \]
and similarly for $\mu_Z$-a.e. $x$
\[ \ld(\mu_Z, x) = \liminf \limits_{n \to \infty} \frac{\#((\N\setminus S) \cap [1,b])}{b} = 0,\  \ud(\mu_Z, x) = \limsup \limits_{n \to \infty} \frac{\#((\N\setminus S) \cap [1,b])}{b} = 1. \]

Let us now prove equalities from point \eqref{it: conv comb process}.  In this case we see that $X^n$ has distribution
\begin{equation}\label{eq: distr nonerg}
	\mu_{X^n} = \frac{1}{2}\mu_Y^{\otimes n} + \frac{1}{2}\mu_Z^{\otimes n}.
\end{equation}
It now follows from \eqref{eq: local dim dyadic} and \eqref{eq: loc dim Y} that for $\mu_{Y}^{\otimes n}$-a.e. $(x_1, \ldots, x_n) \in \R^n$
\[ \ld(\mu_{Y}^{\otimes n}, (x_1, \ldots, x_n)) = \liminf \limits_{b \to \infty} \frac{\sum \limits_{i=1}^n \log \mu_Y(C^1_{b}(x_i))}{b} = n \liminf \limits_{b \to \infty} \frac{\#(S \cap [1,b])}{b} = 0\]
and
\[  \ud(\mu_{Y}^{\otimes n}, (x_1, \ldots, x_n)) = \limsup \limits_{b \to \infty} \frac{\sum \limits_{i=1}^n \log \mu_Y(C^1_{b}(x_i))}{b} = n \limsup \limits_{b \to \infty} \frac{\#(S \cap [1,b])}{b} = n, \]
and similarly for  $\mu_{Z}^{\otimes n}$-a.e. $(x_1, \ldots, x_n) \in \R^n$
\[\ld(\mu_{Z}^{\otimes n}, (x_1, \ldots, x_n)) = 0,\ \ld(\mu_{Z}^{\otimes n}, (x_1, \ldots, x_n)) = n.\]
Combining this with \eqref{eq: ac dim} gives, as $\mu_{Y^n} \ll \mu_{X^n}$ and $\mu_{Z^n} \ll \mu_{X^n}$,
\begin{equation}\label{eq: laldim nonerg}
	\begin{split} \laldim(X^n)& = \int \ld(\mu_{X^n}, (x_1, \ldots, x_n))d\mu_{X^n}(x_1, \ldots, x_n) \\
		& = \frac{1}{2} \int \ld(\mu_{X^n}, (x_1, \ldots, x_n))d\mu_{Y}^{\otimes n}(x_1, \ldots, x_n) + \frac{1}{2} \int \ld(\mu_{X^n}, (x_1, \ldots, x_n))d\mu_{Z}^{\otimes n}(x_1, \ldots, x_n) \\
		& = \frac{1}{2} \int \ld(\mu_{Y}^{\otimes n}, (x_1, \ldots, x_n))d\mu_{Y}^{\otimes n}(x_1, \ldots, x_n) + \frac{1}{2} \int \ld(\mu_{Z}^{\otimes n}, (x_1, \ldots, x_n))d\mu_{Z}^{\otimes n}(x_1, \ldots, x_n)\\
		& = 0
	\end{split}
\end{equation}
and similarly
\begin{equation}\label{eq: ualdim nonerg}
	\ualdim(X^n) = n.
\end{equation}
These give
\[ \mlaldim(\mathbf{X}) =0\ \text{ and }\ \mualdim(\mathbf{X})=1.\]
It remains to prove
\begin{equation}\label{eq: midim idimr nonerg ex} 1/2 = \midim(\mathbf{X}) = \idimr(\mathbf{X}).
\end{equation}
By \eqref{eq: entropy vec part}, \eqref{eq: loc dim Y} and \eqref{eq: loc dim Z}
\[\begin{split}
	H([X^n]_{2^b}|\Delta) & = \frac{1}{2}(H([Y^n]_{2^b}) + H([Z^n]_{2^b})) = 
	\frac{n}{2}(H([Y]_{2^b}) + H([Z]_{2^b})) \\
	& = \frac{n}{2} (\#(S \cap [1,b]) + \#(\N\setminus S) \cap [1,b]) = \frac{nb}{2}.
\end{split}\]
As 
\[ H([X^n]_{2^b}|\Delta) \leq H([X^n]_{2^b}) \leq H(\Delta) + H([X^n]_{2^b}|\Delta) = \log 2 + H([X^n]_{2^b}|\Delta),\]
we obtain \eqref{eq: midim idimr nonerg ex}.

Let us deal now with point \eqref{it: iid process}. In this case, instead of \eqref{eq: distr nonerg}, we have
\begin{equation}\label{eq: distr iid}
	\mu_{X^n} = \left(\frac{1}{2}\mu_Y + \frac{1}{2}\mu_Z\right)^{\otimes n} = \sum \limits_{(\om_1, \ldots \om_n)\in \{0,1\}^n}\ 2^{-n}\bigotimes \limits_{i=1}^n \left( \om_i\mu_Y + (1 - \om_i)\mu_Z\right).
\end{equation}
As the one-dimensional distribution of $\mathbf{X}$ is the same as in the previous point, we see from \eqref{eq: laldim nonerg} and \eqref{eq: ualdim nonerg} that
\[ \laldim(X_1) = 0\ \text{ and }\ \ualdim(X_1)=1.  \]
By Lemmas \ref{lem: aldim mid idimr} and \ref{lem: psi star mlaldim}, it suffices to prove that
\begin{equation}\label{eq: maldim iid}
	\mlaldim(\mathbf{X}) = \mualdim(\mathbf{X}) = 1/2.
\end{equation}
For $\om = (\om_1, \ldots \om_n)\in \{0,1\}^n$ let us denote
\[ \mu_{\om} = \bigotimes \limits_{i=1}^n \left( \om_i\mu_Y + (1 - \om_i)\mu_Z\right).\]
Then by \eqref{eq: loc dim Y} and \eqref{eq: loc dim Z}, for $\mu_\om$-a.e. $(x_1, \ldots, x_n) \in \R^n$
\begin{equation}\label{eq: dim om}
	\begin{split} \ld(\mu_\om, (x_1, \ldots, x_n)) & = \liminf \limits_{b \to \infty} \frac{ \log \prod \limits_{i=1}^n \left(\om_i\mu_Y(C^1_{b}(x_i)) + (1-\om_i)\mu_Z(C^1_{b}(x_i))\right)}{b} \\
		& = \liminf \limits_{b \to \infty} \frac{ \left( \#(S \cap [1,b]) \sum \limits_{i=1}^n \om_i\right) + \left( \#((\N \setminus S) \cap [1,b]) \sum \limits_{i=1}^n (1-\om_i)\right)}{b} \\
		& = \sum \limits_{i=1}^n \om_i + \liminf \limits_{b \to \infty} \frac{ \#((\N \setminus S) \cap [1,b]) \sum \limits_{i=1}^n (1-2\om_i)}{b} \\
		& \geq \sum \limits_{i=1}^n \om_i - \left| \sum \limits_{i=1}^n (1-2\om_i)\right| \\
		& \geq \sum \limits_{i=1}^n \om_i - \left| n - 2\sum \limits_{i=1}^n \om_i\right|.
	\end{split}
\end{equation}
Therefore, by \eqref{eq: ac dim} and \eqref{eq: distr iid}
\[
\begin{split}
	\frac{1}{n}\laldim(X^n) & = \frac{1}{n} \int \ld(\mu_{X^n}, (x_1, \ldots, x_n))d\mu_{X^n}(x_1, \ldots, x_n) \\
	& = \sum \limits_{\om = (\om_1, \ldots \om_n)\in \{0,1\}^n}\ \frac{2^{-n}}{n} \int \ld(\mu_{X^n}, (x_1, \ldots, x_n))d\mu_{\om}(x_1, \ldots, x_n) \\
	& = \sum \limits_{\om = (\om_1, \ldots \om_n)\in \{0,1\}^n}\ \frac{2^{-n}}{n} \int \ld(\mu_{\om}, (x_1, \ldots, x_n))d\mu_{\om}(x_1, \ldots, x_n) \\
	& \geq  \sum \limits_{\om = (\om_1, \ldots \om_n)\in \{0,1\}^n}\ 2^{-n} \left( \frac{1}{n}\sum \limits_{i=1}^n \om_i - \left| 1 - \frac{2}{n}\sum \limits_{i=1}^n \om_i\right| \right) \\
	& = \mathbb{E} \left( \frac{1}{n}\sum \limits_{i=1}^n \Om_i - \left| 1 - \frac{2}{n}\sum \limits_{i=1}^n \Om_i\right| \right),
\end{split}
\]
where $\Om_1, \Om_2, \ldots$ is a sequence of i.i.d. random variables such that $\PP(\Om_i = 0) = \PP(\Om_i = 1) = 1/2.$ As $\lim \limits_{n \to \infty} \frac{1}{n} \sum \limits_{i=1}^n \Om_i = 1/2$ almost surely, we conclude that
\[ \mlaldim(\mathbf{X}) = \liminf \limits_{n\to \infty}\frac{1}{n}\laldim(X^n) \geq 1/2. \]
Similarly as in \eqref{eq: dim om}, we can also prove $\ud(\mu_\om, (x_1, \ldots, x_n)) \leq \sum \limits_{i=1}^n\om_i + \left|n - 2\sum \limits_{i=1}^n \om_i \right|$ for $\mu_\om$-a.e. $(x_1, \ldots, x_n) \in \R^n$. Consequently  $\mualdim(\mathbf{X}) \leq 1/2$, establishing \eqref{eq: maldim iid}.

\bibliographystyle{alpha}
\bibliography{universal_bib}

\end{document}